\documentclass[10pt]{article}

\usepackage{color}
\usepackage{ifthen}
\usepackage[utf8]{inputenc}
\usepackage[english]{babel}
\usepackage{amsmath,amsfonts,amssymb,amsthm}
\usepackage{bbm}
\usepackage{lmodern}
\usepackage[scaled]{helvet}
\usepackage[T1]{fontenc}
\usepackage{lettrine} 
\usepackage{graphicx,epstopdf,color}
\usepackage{booktabs}
\usepackage{authblk}
\usepackage{fullpage}
\usepackage[colorlinks=true,citecolor=blue,menucolor=black,linkcolor=blue]{hyperref}
\usepackage{float}
\usepackage{subcaption}
\usepackage{parskip}

\begingroup
    \makeatletter
    \@for\theoremstyle:=definition,remark,plain\do{%
        \expandafter\g@addto@macro\csname th@\theoremstyle\endcsname{%
            \addtolength\thm@preskip\parskip
            }%
        }
\endgroup

  \newcommand{\R}{\ensuremath{\mathbb{R}}}
  \newcommand{\E}{\ensuremath{\mathbb{E}}}
  \newcommand{\C}{\ensuremath{\mathbb{C}}}
  \newcommand{\I}{\ensuremath{\mathbb{I}}}
  \newcommand{\Z}{\ensuremath{\mathbb{Z}}}

  \newcommand{\Pc}{\mathcal{P}}

  \newcommand{\Hc}{\mathcal{H}}
  \newcommand{\Wc}{\mathcal{W}}
  \newcommand{\Vc}{\mathcal{V}}
  \newcommand{\Ec}{\mathcal{E}}
  
\newcommand{\tr}{\text{tr }}

\newcommand{\V}[1]{\ensuremath{\mathbf{#1}}}

\newcommand{\norm}[1]{\left|\left| #1 \right|\right|}

\newcommand{\abs}[1]{\left \vert #1 \right \vert}

\newcommand{\TODO}[1]{ 
  \ifx\NOTES\undefined\else
	{\tt \color{red} [TODO:#1] } 
  \fi
}

\newcommand{\one}{\ensuremath{\mathbbm{1}}}

\newcommand{\NOTE}[1]{ 

\ifx\NOTES\undefined\else
  \footnote{ {\color{blue} NOTE: #1}}  
\fi
}

\newcommand{\ecomment}[1]{ 
\ifx\NOTES\undefined\else 
{\color{blue}[E]}\footnote{ {\color{blue} Eran: #1}}
\fi
}

\newcommand{\mcomment}[1]{ 
\ifx\NOTES\undefined\else
  {\color{green} [M]}\footnote{ {\color{green} Matan: #1}}  
\fi
}

\newcommand{\aslim}{\stackrel{a.s.}{\rightarrow}}
\newcommand{\aseq}{\stackrel{a.s.}{=\joinrel =}}

\newcommand{\iid}{\stackrel{\text{iid}}{\sim}}

\newcommand{\Pd}[3]{\ifthenelse{\equal{#3}{1}}{\frac{\partial #1}{\partial #2}}{\frac{\partial^{#3} #1}{\partial #2^{#3}}}}

\theoremstyle{plain}
\newtheorem{theorem}{Theorem}
\newtheorem{proposition}{Proposition}

\newtheorem{lemma}{Lemma}

\theoremstyle{definition}
\newtheorem{definition}{Definition}
\newtheorem{assumption}{Assumption}

\theoremstyle{remark}
\newtheorem{remark}{Remark}
\newtheorem{example}{Example}

\usepackage{lineno}

\newcommand{\Gr}{\ensuremath{\mathbb{G}}}

\title{The Noise-Sensitivity Phase Transition in Spectral Group
	Synchronization  Over Compact Groups}
\date{}
\author[1]{Elad Romanov \thanks{E-mail: elad.romanov@mail.huji.ac.il}}
\author[1]{Matan Gavish \thanks{E-mail: gavish@cs.huji.ac.il}} 

\affil[1]{School of Computer Science and Engineering, The Hebrew University, 
Jerusalem, Israel}

\def \NOTES{}

\begin{document}
	
	\maketitle

	\begin{abstract}

	In Group Synchronization, one attempts to find a collection of unknown group
	elements from noisy measurements of their pairwise differences.  Several
	important problems in vision and data analysis reduce to group
	synchronization over various compact groups.  Spectral Group Synchronization
	is a commonly used, robust algorithm for solving group synchronization
	problems, which relies on diagonalization of a block matrix whose blocks are
	matrix representations of the measured pairwise differences.  Assuming
	uniformly distributed measurement errors, we present a rigorous analysis of
	the accuracy and noise sensitivity of spectral group synchronization
	algorithms over any compact group, up to the rounding error.
	We identify a Baik-Ben Arous-P\'ech\'e
	type phase transition in the noise level, beyond which spectral group
	synchronization necessarily fails.  Below the phase transition, spectral
	group synchronization succeeds in recovering the unknown group elements, but
	its performance deteriorates with the noise level. We provide asymptotically
	exact formulas for the accuracy of spectral group synchronization 
	below the phase transition, up to the rounding error.
	We also provide a consistent risk estimate, allowing practitioners to
	estimate the method's accuracy from available measurements.

%
	\end{abstract}

	{\bf Acknowledgments.}
	We thank Yoel Shkolnisky and Amit Singer for helpful discussions.
	We thank Nati Linial for random help on random graphs.
	We are indebted to the anonymous reviewers for their many valuable comments and suggestions.
	ER was partially supported by 
	Israeli Science Foundation grant no. 1523/16 and the HUJI Leibniz center.
	ER and MG were partially supported by 
	United States – Israel Binational Science Foundation grant no. 2016201.

	{\bf Reproducibility advisory.}
	All the experimental results cited in this paper are fully reproducible.
	Code to generate all the figures included in this paper and their underlying
	data can be found in
	the code supplement
	\cite{CODE}.

	\newpage
	\tableofcontents

	\section{Introduction} \label{sec:intro}


	In group synchronization, one attempts to recover unknown group elements
	$g_1,\ldots,g_n\in \Gr$ from incomplete, noisy measurements 
	of their group differences $\{g_i g_j^{-1}\}_{(i,j)\in \Lambda}$, with 
	$\Lambda\subset \left\{ (i,j) \,|\,1\leq i<j\leq n \right\}$ the set of available
	pairwise difference measurements.
	%
	Numerous problems in signal processing, computer vision and machine learning
	can be cast as group synchronization problems for an appropriate choice of
	group $\Gr$. Examples include molecular structure determination by
	Cryo-electron microscopy \cite{singer2011three} 
	(synchronization over the special orthogonal group $SO(3)$); determination of
	Structure From Motion  in computer vision \cite{tron2016survey,
	arrigoni2015spectral,bernard2015solution}
	(synchronization over $SO(3)$);
	pose graph estimation \cite{carlone2015initialization} (synchronization over 
	$SO(3)$); sensor
	localization \cite{cucuringu2012sensor,peters2015sensor,tron2009distributed} (synchronization over 
	the special Euclidean group $SE(d)$);
	community detection in graphs
	\cite{cucuringu2015synchronization} (synchronization over 
	$\mathbb{Z}_2$); ranking \cite{cucuringu2016sync} (synchronization over $SO(2)$); multireference
	alignment in signal processing \cite{bandeira2014multireference}
	(synchronization over the finite
	cyclic group $\mathbb{Z}_p$); global alignment in dynamical systems
	\cite{sonday2013noisy} and network clock synchronization
	\cite{giridhar2006distributed} (synchronization over $\R$
	or a large finite cyclic group). 

Group synchronization is fundamentally an {\em integration} problem, in which one
	attempts to recover a global structure from measured local
	differences. However, unlike calculus integration,  in which 
	the local differences are
	measured on a grid in $\R^d$, in group synchronization
	the differences are measured along the edges of a graph,
	whose nodes are the group elements $g_1,\ldots,g_n$ and whose edges  
	correspond to the set of available difference measurements 
	$\Lambda\subset \left\{ i,j \,|\,1\leq i< j\leq n \right\}$.
	It is instructive to think about group synchronization as a generalization 
	of the classical  problem of 
	{\em global positioning
	from local distances} problem \cite{young1938discussion}, in which
	one attempts to recover as unknown point cloud
	$\V{x}_1,\ldots,\V{x}_n\in\R^d$ from their pairwise
	Euclidean distances 
	$d_{i,j}=\norm{\V{x}_i-\V{x}_j}^2$. Indeed, if instead of the distances $d_{i,j}$ we were given the actual displacements 
	$\V{x}_i-\V{x}_j$, then the problem would correspond exactly to synchronization over the (non-compact) group $\R^d$ (see, e.g, \cite{barooah2007estimation}). 
	
	In the absence of noise, namely, when the set of available measurements is
	exactly $\{g_i g_j^{-1}\}_{(i,j)\in \Lambda}$,  the group synchronization problem
	can be solved if and only if the graph on ${g_1,\ldots,g_n}$ with edge set
	$\Lambda$ is connected. Indeed, if the graph is connected, the group elements
	$g_1,\ldots,g_n$  are recovered by traversing any spanning tree of edges in
	the graph. Note that the problem is only defined up to a global
	``phase'', or global group element, in the sense that if $g_1,\ldots,g_n$ is
	a solution to the problem, then so is $a g_1,\ldots, a g_n$, for any
	$a\in\Gr$. Conversely, if the graph is not connected, the solution is
	determined up to such a phase in each connected component separately,
	and is therefore ill-posed. 
	
	However, when the difference measurements $g_i g_j^{-1}$ are contaminated by
	measurement noise, it is not immediately clear how to proceed. The
	problem then becomes important from a practical perspective and interesting from
	a theoretical perspective. Indeed, it has received considerable attention
	recently: various algorithms have been proposed,
	some based on semidefinite programming \cite{singer2011angular,wang2013exact}, 
	maximum likelihood \cite{boumal2013robust,bandeira2017tightness}, non-unique games
	\cite{bandeira2015non} and message passing in graphs
	\cite{perry2016message}.
	
	In this paper, we focus on the spectral method for group synchronization
	\cite{singer2008remark,singer2011angular}, or
	{\em Spectral Group Synchronization} for short (the work \cite{yu2012angular} is also related).  This method can be used when
	the group $\Gr$  admits a
	finite-dimensional, faithful, unitary representation $\pi:\Gr\to U(d)$. As described in more detail below, 
	this method proceeds by extracting the desired group elements $g_1,\ldots,g_n$ 
	from the top eigenvectors of the Hermitian 
	$nd$-by-$nd$ block matrix $Y$, whose $i,j$-th
	block is $\pi(g_i g_j^{-1})$. 
	It is easy to implement, and basically reduces to power iterations, 
	making it quite feasible and scalable in practice. 
	We note, however, that the SDP instances arising from group synchronization are known to be
	efficiently solvable using Burer-Monteiro-type relaxations \cite{boumal2016burermonteiro}, see also \cite{bandeira2017tightness,rosen2019se}. 
	Alternatively, in some cases where the power method is used to compute the top eigenvectors, 
	generalized power methods (where in each iteration, one also performs a projection step onto the group) may also be applicable, see for example \cite{zhong2018near}. Another interesting variation is \cite{carmona2013analytical}, which essentially considers the spectral method for $SO(3)$ synchronization using quaternions for the representation. 
	That being said, the spectral method still has the upper-hand in terms of simplicity, making it popular among practitioners.
	 We remark that spectral methods for integration of a global structure
	 from noisy, local difference measurements have a time-honored tradition.
	 Indeed, almost all Euclidean data embedding techniques, from
	 Multidimensional Scaling \cite{Torgerson1952a,gower1966some} to 
	  manifold learning methods such as  
	  \cite{tenenbaum2000Isomap,Belkin2001LaplacianEigMaps,coifman2006diffusion},
	 basically reconstruct a global Euclidean structure using the eigenvectors of a local
	 difference matrix. 

	 Spectral group synchronization can
	 always be used when $\Gr$ is compact; recently,
	 \cite{ozyesil2016synchronization} proposed a compactification scheme that
	 allows it to be used for non-compact group such as the special Euclidean
	 (rigid motion) group. For the remainder of this paper, we will assume that
	 synchronization takes place over an arbitrary compact group $\Gr$.

	 Clearly, as the measurement noise level rises, any method for group
	 synchronization over $\Gr$ should suffer a performance loss, and, possibly, even break down completely
	 once the noise level exceeds some threshold. Practitioners have commented that spectral group
	 synchronization is quite robust to noise; still, and despite its widespread
	 use, the literature currently does not offer a systematic treatment of its
	 noise sensitivity.
	 
	 There are several possible ways to model measurement noise. Here, we adopt
	 the noise model proposed by  
	 \cite{singer2011angular}. In this model, each of the measurements 
	 $\{g_i g_j^{-1}\}_{(i,j)\in \Lambda}$ is exact with equal probability $p$, 
	 and is corrupted with probability $1-p$. Corrupt measurements are assumed 
	 to be distributed uniformly on $\Gr$, namely, sampled from its 
	 Haar measure. Under this model,
	 \cite{singer2011angular} has proposed a strategy
	 for formal analysis of spectral group synchronization over
	 $SO(2)$, based on ideas from random matrix theory.
	 This paper implements that strategy and presents a systematic and fully 
	 rigorous analysis of the noise sensitivity of spectral group
	 synchronization over an arbitrary compact group $\Gr$. We show that, as
	 predicted by 
 \cite{singer2011angular,tzeneva2011global,boumal2014optimization},
 spectral group synchronization
	 exhibits an asymptotically sharp phase transition similar to the 
	  Baik-Ben Arous-P\'ech\'e phase transition \cite{baik2005phase}. This phase
	  transition coincides with the breakdown point of  spectral
	  group synchronization, namely, with the critical 
	  noise level, beyond which it necessarily
	  fails to recover the group elements $g_1,\ldots,g_n$. We further 
	  combine 
	recent techniques from random
	matrix theory and elementary facts on group representations to derive
	rigorous and asymptotically exact results on the noise sensitivity 
	of spectral group
	synchronization (up to the rounding error\footnote{	As discussed in more detail below, spectral group synchronization concludes
	with a {\em rounding} step, in which group elements are identified by
	rounding each $d$-by-$d$ block, obtained from the $d$ top 
	eigenvectors of $Y$, to the
	nearest matrix in the representation $\pi(\Gr)$. Error analysis of this step
is necessarily group-specific and is beyond our present scope.}) when the noise level is below the critical threshold.

	\section{Notation and setup} \label{sec:setup}

	\paragraph{Observations.}
	Let $\Gr$ be a compact group, equipped with its normalized Haar measure, and some
	$d$-dimensional, non-trivial, faithful (that is, no two different group elements are mapped into the same matrix),
	unitary representation $\pi : \Gr \to U(d)$.
	Assume that  $g_1,\ldots,g_n$ are 
	$n$ unknown group elements to be recovered. As $\pi$ is faithful\footnote{Our results hold just as well when the representation $\pi$ is not faithful. In that case, since we only interact (in terms of the measurements and the fidelity metric) with the group $\Gr$ via the representation $\pi$, the problem could actually be thought of as synchronization over $\Gr/ Ker(\pi)$, where $Ker(\pi)=\left\{g \,:\,\pi(g)=Id \right\}$.}, the
	problem reduces to recovery of 
		their images $\pi(g_1),\ldots,\pi(g_n) \in U(d)$.
%
	%
		Write $[n]=\left\{ 1,\ldots,n \right\}$ and let
	$\Lambda$ be an \emph{undirected} graph whose vertex set is $[n]$ (we don't allow self-loops). With a slight abuse of notation, by saying $(i,j)\in\Lambda$ we mean that there is an edge connecting $i$ and $j$, and this is exactly the same as saying $(j,i)\in\Lambda$ (that is, the order of $i$ and $j$ doesn't matter).
Let $\{ g_{i,j}\}_{(i,j)\in\Lambda} \subset \Gr $ denote the measured group
differences. The task at hand
is to recover $\pi(g_1),\ldots,\pi(g_n)$ from the available observations 
$\{ g_{i,j}\}_{(i,j)\in\Lambda} $.
%
\paragraph{Noise model.}	Following \cite{singer2011angular}, we consider the case when $\Lambda$ is a
	random Erd\H{o}s--R\'enyi graph, where each edge appears independently
	with some probability $q$, and where
	the corrupted measurements are chosen uniformly at random.  
	More specifically, conditioned on $(i,j)\in \Lambda$, sticking to the convention $i<j$, with
	probability $p$ we measure the real group difference 
	$g_{ij} = g_i g_j^{-1}$. 
	Otherwise, we measure a random group element, in the sense that 
	$g_{ij}\iid \textrm{Haar}$ 
	is sampled from the normalized Haar measure on $\Gr$ \footnote{We provide concise
		background on the Haar measure and other related group-theoretic
notions in the Appendix.}. We will sometimes refer to this form of noise by "outlier-type" corruption. 
We also consider an extension of this noise model, where each observed measurement is also corrupted by some \emph{additive} noise. That is, conditioned on $(i,j)\in\Lambda$, instead of being given $\pi(g_{ij})$ (where $g_{ij}$ is either the true difference or a random group element, as before), we measure a matrix of the form $\pi(g_{ij})+\epsilon_{ij}$, where $\epsilon_{ij}$ is a $d\times d$ i.i.d noise matrix, each of whose entries is Gaussian (real or complex \footnote{Recall that a complex Gaussian random variable with mean $0$ and variance $1$ is of the form $X+iY$, where $X,Y\sim N(0,1/2)$}) with mean $0$ and variance (absolute second moment) $\sigma/\sqrt{d}$ (the normalization by $\sqrt{d}$ is convenient). In the sequel, we will compute asymptotic results as the number of measurements $n\to\infty$. We will allow the probabilities $p=p_n$ and $q=q_n$ and additive noise intensity $\sigma=\sigma_n$ to vary with $n$, at a rate which will be made precise later.
 

	\paragraph{The Spectral Group Synchronization method.}
	Introduced in \cite{singer2011angular}, the method proceeds as
	follows. Define the 
	$nd$-by-$nd$ Hermitian block matrix $Y$ with $d$-by-$d$ blocks
	$y_{ij}=\pi(g_{ij})$,
for $i<j$, and of course $y_{ji}=y_{ij}^*$. On the block-diagonal, we always have $y_{ii}=I$ (the $d\times d$ identity matrix). 
%
%
Here, $g_{ij}$ is the measurement corresponding to the $(i,j)$ edge, 
and by convention $y_{ij}=0$ if $(i,j)\notin \Lambda$.
Observe that in the fully-observed, noiseless case, namely the case where
$q=1$, $p=1$ and $\sigma=0$, $Y$ is a rank-$d$ matrix that decomposes as $Y=XX^{*}$ where $X \in \C^{nd \times d}$ is the block matrix
	\[
	X = \begin{bmatrix}
		\pi(g_1) \\
		\vdots \\
		\pi(g_n)
	\end{bmatrix} \,.
	\]
	Since
	the blocks $\pi(g_{i})$ are all unitary matrices, 
	the columns of $X$ are orthogonal vectors with norm $\sqrt{n}$. This make
$Y=XX^{*}$ an eigen-decomposition of $Y$, where there is a single
$d$-dimensional 
eigenspace of dimension $d$, corresponding to the eigenvalue $n$ and spanned by the columns of $X$. 
This suggests that even in the noisy case, the top $d$ eigenvectors of $Y$ should
approximate $X$ well. 
	
Let $\tilde{X}$ be the $nd$-by-$d$ 
matrix whose columns are the top $d$ eigenvectors
of $Y$. We think of $\tilde{X}$ as a block matrix with a single column of $n$
blocks, each $d$-by-$d$. In the noisy case, these blocks do not necessarily
correspond to elements of the representation $\pi$. 
Ideally, the spectral method should therefore conclude
with a {\em rounding} step, 
in which we produce actual group elements $\hat{g}_1,\ldots,\hat{g}_n$ from $\tilde{X}$.
Denote by $\hat{X} \in \C^{nd\times d}$ the block matrix 
	\[
	\hat{X} = \begin{bmatrix}
		\pi(\hat{g}_1) \\
		\vdots \\
		\pi(\hat{g}_n)
	\end{bmatrix} \,,
	\]
	(as $X$ was for $g_1,\ldots,g_n$). As discussed below, 
	the choice of rounding algorithm -- an important component of the overall
	synchronization method -- depends on the specific group at hand.

\paragraph{Reconstruction quality.}
We measure the quality of our reconstruction $\hat{g}_1,\ldots,\hat{g}_n$ using  the \emph{average squared alignment error},
\begin{equation}
\label{eq:mse_def}
\begin{split}
	MSE(X,\hat{X}) 
	&= \frac{1}{n^2}\sum_{i,j=1}^n \norm{\pi(g_i g_j^{-1}) - \pi(\hat{g}_i \hat{g}_j^{-1})}_F^2 \\
	&= \norm{ \frac{1}{n} XX^{*} - \frac{1}{n}\hat{X}\hat{X}^* }_F ^2 \,.
\end{split}
\end{equation}
Note that the error $MSE(X,\hat{X})$ cannot be directly computed from the data, as some of the real pairwise differences $g_i g_j^{-1}$ may be corrupted or missing.

\paragraph{MSE proxy.}
To the best of our knowledge, it is very difficult to give any meaningful bounds on the performance of any computationally feasible rounding procedure, beyond very limited special cases (most prominently, $\Gr = U(1)=\left\{z\in\C\,:\,\abs{z}=1 \right\}$). We mention two natural rounding strategies in Subsection \ref{subsection:rounding} below. 
To circumvent this difficulty, and keep the discussion in the generality of an arbitrary compact group, we shall assume that we have access to some rounding procedure that produces a rounding error
\[
R(\hat{X},\tilde{X}) =
\norm{\frac{1}{n}\hat{X}\hat{X}^*-\tilde{X}\tilde{X}}^2_F\,,
\]
which is negligible with respect to the overall MSE.
Specifically, we will assume that 
the error $MSE(X,\hat{X})$ is dominated by the \emph{MSE proxy},
\begin{equation}
\begin{split}
\overline{MSE}(X,\tilde{X}) 
&= \norm{\frac{1}{n}XX^*-\tilde{X}\tilde{X}^{*}}_F^2 \\
&= 2d - 2 \tr \left( (XX^{*}/n)(\tilde{X}\tilde{X}^*) \right) \,,
\end{split}
\end{equation}
which essentially measures the degree to which the space spanned by the top $d$
eigenvalues of $Y$ aligns with the column span of $X$. 
In Section \ref{subsect:experiment5} we report numerical evidence that the
rounding error is 
indeed negligible for $\Gr=\mathbb{Z}_2$ and for $\Gr=O(3)$, each with a
straightforward rounding scheme. 
Since
\[
\sqrt{\overline{MSE}}-\sqrt{R} \le \sqrt{MSE} \le
\sqrt{\overline{MSE}}+\sqrt{R} \,,
\] 
under this assumption it is enough to study the MSE proxy.

		\paragraph{Contributions.}
Our main result is the precise asymptotic behavior 	of the MSE proxy
	in the large sample limit $n\to\infty$, where we consider a sequence of probabilities $p_n$, $q_n$, and an additive noise intensity $\sigma_n$, all depending on $n$. We assume that our measurement graph $\Lambda$ is \emph{dense}, with an average degree of order $\Omega\left(\log^c(n)\right)$, where $c>1$ is sufficiently large.  
	(formally, as already mentioned, we assume that the graph of available measurements is an Erd\H{o}s--R\'enyi 
	graph where every edge appears with probability $q_n$). 
	Generalization of
	our analysis to the case of a sparser random graph (say $q_n =
	a\log(n)/n$, for a large number $a>1$, which would already guarantee that the graph $\Lambda$ is connected with an overwhelming probability) seems to be currently out of the reach for our tools. To be more precise, we will require,
	\begin{assumption}
	\label{assum:prob}
		The sequence of probabilities $q_n$ and $p_n$ satisfy $q_n p_n = \omega \left( \log^{c}(n)/n \right)$, where $c>1$ is some universal constant, which will not dependent on $\Gr$ or the choice of representation $\pi$ (or on $d$). In particular, for any $\delta>0$, the condition $p_n q_n = \Omega\left(n^{-1+\delta}\right)$ already suffices. Trying to optimize for the best $c>1$ that could possibly be guaranteed by our technique appears to be a somewhat messy task of very little interest, and we do not attempt to undertake it (it requires, for a start, more precise estimates for a certain combinatorial calculation that we use in the proof).
	\end{assumption}
	As for the representation $\pi$, we will assume throughout the rest of this paper,
	\begin{assumption}
	\label{assum:irreducible}
		The representation $\pi : \Gr \to U(d)$ is non-trivial and irreducible.
	\end{assumption}

	Under these assumptions, in the case where there is no additive noise ($\sigma_n=0$), we observe a recoverability phase-transition at
	scale $p_n \sqrt{q_n} \sim 1/\sqrt{n}$. In particular, when $q_n=\Omega(1)$, our
	results imply that the spectral method is remarkably robust to an order of
	$\sim (1-1/\sqrt{n})\abs{\Lambda}$ outlier-type corruptions, a fact already observed, in
	the case where $\Gr$ is the unit circle, in the original analysis of \cite{singer2011angular}. 


	\subsection{Some remaks on the rounding step}
	\label{subsection:rounding}
	
	We identify two sensible rounding strategies to produce our estimates $\hat{g}_1,\ldots,\hat{g}_n$ from the matrix of $d$ top eigenvectors, $\tilde{X}$:
	
	\paragraph{Ideal rounding.} Considering our error criterion \eqref{eq:mse_def}, the natural rounding procedure would be to take
	\begin{equation}
	\label{eq:rounding_ideal}
	\hat{g}_1,\ldots,\hat{g}_n \in \arg\min \norm{\frac{1}{n}\hat{X}\hat{X}^{*}-\tilde{X}\tilde{X}^{*}}_F^{2} \,,
	\end{equation}
	where recall that $\hat{X}$ is a column block matrix whose blocks are $\pi(\hat{g}_i)$. 
	This rounding procedure bears a strong resemblance to the \emph{least squared error} estimator,
	\begin{equation*}
	\begin{split}
	\hat{g}_1^{LSE},\ldots,\hat{g}_{n}^{LSE} 
	&\in \arg\min \sum_{i<j\,:\,(i,j)\in \Lambda} \norm{\pi(\hat{g}_i \hat{g}^{-1}_j) - y_{ij}}_F^2 \,.
	\end{split}
	\end{equation*}
	Unfortunately, ideal rounding can be, in general, a computationally hard problem. For instance, when $\Gr = \mathbb{Z}_2$, the ideal rounding problem basically amounts to solving 
	\[
	\arg\max_{g_1,\ldots,g_n\in \left\{ \pm 1 \right\}} \tr (W \hat{X}\hat{X}^{*}) = \arg\max_{g_1,\ldots,g_n\in \left\{ \pm 1 \right\}} \sum_{i,j} w_{ij}g_i g_j
	\,,
	\]
	where $W=\tilde{X}\tilde{X}^{*}$.
	This essentially amounts to solving an instance of the MAX-CUT problem, which is generally known to be NP-complete.
	
	While ideal rounding is generally computationally infeasible, it is very easy to give a loose upper bound on the true MSE one gets, in terms of the MSE proxy: 
	\begin{equation}\label{eq:ideal_rounding_bound}
	\begin{split}
	MSE(X,\hat{X}) 
	&= \norm{\frac{1}{n}XX^*-\frac{1}{n}\hat{X}\hat{X}^*}_F^2 \\
	&\le \left( \norm{\frac{1}{n}XX^*-\tilde{X}\tilde{X}^*}_F + \norm{\frac{1}{n}\hat{X}\hat{X}^*-\tilde{X}\tilde{X}^*}_F \right)^2 \\
	&\le 4 \norm{\frac{1}{n}XX^{*}-\tilde{X}\tilde{X}^*}_F^2 = 4\cdot\overline{MSE}(X,\tilde{X}) \,,
	\end{split}
	\end{equation}
	since for ideal rounding, $\norm{\frac{1}{n}\hat{X}\hat{X}^*-\tilde{X}\tilde{X}^*}_F \le \norm{\frac{1}{n}XX^*-\tilde{X}\tilde{X}^*}_F$.

	\paragraph{Block-wise rounding.}
	Since in the noiseless case, the columns
	of $X$ \emph{are} the eigenvectors of $Y$, it makes sense 
	to round the matrix $\tilde{X}$ directly. That is, we take
	\begin{equation}
	\label{eq:roudning_block}
	\hat{g}_1,\ldots,\hat{g}_n \in \arg\min \norm{\hat{X} -
		\sqrt{n}\tilde{X}}_F \,. 
	\end{equation}  
	While still group-specific,  this is usually an easy problem; indeed, notice that the problem decouples across the optimization variables, so that
	\begin{equation}
	\hat{g}_i \in \arg\min_{g\in\Gr} \norm{\pi(g)-\sqrt{n}\tilde{X}_i}_F \,,	
	\end{equation}
	where $\tilde{X}_i \in \C^{d\times d}$ is the $i$-th block of $\tilde{X}$. It is \emph{hoped} (but not guaranteed) that the eigenvectors $\tilde{X}$ are sufficiently delocalized (and that the image of $\pi$ is sufficiently dense in $U(d)$ \footnote{That is, for an arbitrary $U\in U(d)$ we could find some $g\in\Gr$ such that $\norm{\pi(g)-U}_F$ is not too large.}) so that we don't lose much by decoupling the ideal rounding problem \eqref{eq:rounding_ideal}.
	Actually proving a theorem along these lines looks, to us, to be a non-trivial task.
	
	Denote by 
	\[
	D_i = \norm{\pi(\hat{g}_i)-\sqrt{n}\tilde{X}_i}_F
	\]
	the rounding error for the $i$-th reconstructed group element. We may bound, as before (Eq. \eqref{eq:ideal_rounding_bound}),
	\begin{align*}
		MSE(X,\hat{X}) 
		&\le 2\overline{MSE}(X,\tilde{X}) +  \frac{2}{n^2}\sum_{i \ne j} \norm{\pi(\hat{g}_i)\pi(\hat{g}_j)^*-n\tilde{X_i}\tilde{X_j}^*}_F^2 \,, 
	\end{align*}
	using
	\begin{align*}
		\norm{\pi(\hat{g}_i)\pi(\hat{g}_j)^*-n\tilde{X_i}\tilde{X_j}^*}_F^2 
		&\le 2 \norm{\pi(\hat{g}_i)\pi(\hat{g}_j)^*-\sqrt{n}\pi(\hat{g}_i)\tilde{X_j}^*}_F^2 + 2 \norm{\sqrt{n}\pi(\hat{g}_i)\tilde{X_j}^*-n\tilde{X_i}\tilde{X_j}^*}_F^2 \\
		&\le 2 \norm{\pi(\hat{g}_i)}^2\norm{\pi(\hat{g}_j)^*-\sqrt{n}\tilde{X_j}^*}_F^2 + 2\norm{\sqrt{n}\tilde{X_j}^*}^2 \norm{\pi(\hat{g}_i)-\sqrt{n}\tilde{X_i}}_F^2\\
		&\le 2\left(1 + \norm{\sqrt{n}\tilde{X_j}^*}^2_F\right) D_i^2 \,,
	\end{align*}
	and that
	\[
	\sum_{j=1}^{n} \norm{\sqrt{n}\tilde{X_j}^*}_F^2 = \norm{\sqrt{n}\tilde{X}^*}_F^2 = nd\,,
	\]
	we get 
	\[
	MSE(X,\hat{X}) \le 2\overline{MSE}(X,\tilde{X}) + \frac{4(d+1)}{n}\sum_{i=1}^{n}D_i^2 \,,
	\]
	so that if the average block-wise rounding error $\frac{1}{n}\sum_{i=1}^n D_i^2$ is small, then a bound on $\overline{MSE}$ also yields a reasonable upper bound on $MSE$. We do not, however, know how to say anything meaningful on $\frac{1}{n}\sum_{i=1}^n D_i^2$. 
	
	The main difficulty here, we believe, lies with the fact that when $d>1$, block-wise rounding is, really, no longer the "obvious" thing to do, because of the following fundamental ambiguity: \emph{there is no "intrinsic", canonical, basis} for the principal subspace of the "noiseless" data matrix $XX^*$. That is, in terms of the spectral method, any other basis has just as high a standing as the one basis we are \emph{really} interested in - the one given by the columns of $X$. All of these bases, of course, relate to one another through an isometry, and in the case $d=1$, this means simply multiplication by any complex number of modulus $1$ \footnote{In the case where the representation is real, we may work exclusively with real numbers. In that case, the isometry is simply multiplication by $\{\pm 1\}$. }. In the special case of the group $\Gr=U(1)$, the spectral method essentially \emph{trully} recovers the signal, up to a global element! This fact was used, e.g, in the analyses of \cite{boumal2016nonconvex,liu2017estimation}. Note that this line of reasoning already break down when $d>1$ and $\Gr=U(d)$, in the sense that the notion of recovery up to global phase no longer captures the permissible isometries of the principal subspace: On the one hand, the group of isometries is given by all the matrices 
	\[
	\left\{ XUX^* \,:\,U\in U(d)\right\} \subset U(nd) \,,
	\]
	whereas alignment by a global group element corresponds to multiplication by any $nd\times nd$ block diagonal matrix, whose block diagonal consists of $n$ copies of some single $U\in U(d)$. It is only when $d=1$ that these two groups coincide!

	\section{Main results}
	\begin{theorem}[Limiting eigenvector correlations]\label{thm:limiting_mse}
		There is a numerical constant $c>1$ such that the following holds. 
		Suppose that $p_n q_n = \omega\left(\log^c(n)/n\right)$. Consider
		\[
		\beta_n = \frac{1}{p_n\sqrt{q_n n}} \sqrt{1-p_n+\sigma_n^2} \,,
		\]
		and suppose that $\beta_n \to 1/\gamma \in [0,\infty]$ as $n\to\infty$. 
		Then under the probabilistic model described above, 
		\begin{equation}
			\lim_{n\to\infty} \overline{MSE}(X,\tilde{X}) \aseq \begin{cases}
		\frac{2d}{\gamma^2}, &\text{ if } \gamma > 1\\
		2d, &\text{ otherwise }
		\end{cases}\,.
		\end{equation}
	\end{theorem}

	Theorem \ref{thm:limiting_mse} demonstrates that, in particular, recovery by the spectral method exhibits a \emph{phase-transition}: when the effective signal level is $\gamma \le 1$, the top eigenvectors of the measurement matrix $Y$ are completely uncorrelated with the columns of $X$. \footnote{Note that the reason why the limiting expression for $\overline{MSE}$ is proportional to $d$, owes to the fact that we are measuring the correlation between two $d$-dimensional subspaces. Indeed, when the columns of $X$ and $\tilde{X}$ are completely uncorrelated, we have $\overline{MSE}(X,\tilde{X})=\frac{1}{n}\norm{XX^*}_F^2 + \norm{\tilde{X}\tilde{X}^*}_F^2= d+d=2d$. } 


	\begin{theorem}[Interpretation of observed eigenvalues]\label{thm:limiting_statistics}
		Suppose that we work under the setting of Theorem \ref{thm:limiting_mse}.
		Define the statistic
		\begin{equation}
			\hat{\phi}(Y) = \frac{2d}{\left( \eta + \sqrt{\eta^2 - 1} \right)^2 }
		\end{equation}
		where $\eta=\frac{\lambda_{1}}{\lambda_{d+1}} $ and $\lambda_1 \ge \lambda_2 \ge \ldots \ge \lambda_{nd}$ are the observed eigenvalues of the measured matrix $Y$. Then  $\hat{\phi}$ is an asymptotically strongly consistent estimate of the MSE proxy, in the sense that
		\begin{equation}
			\overline{MSE}(X,\tilde{X}) - \hat{\phi}(Y) \to 0
		\end{equation}
		almost surely as $n\to\infty$.

	\end{theorem}	


	\begin{remark}
		\label{rem:no_additive_noise}
		Let us consider, for a moment, the case where there is no
		additive noise, that is, $\sigma_n=0$, and suppose that $\beta_n
		\to 1/\gamma$ for $\gamma\in(0,\infty)$. Under the condition
		that $p_nq_n = \omega(\log^c(n)/n)$, the case
		$\lim\sup_{n\to\infty} p_n > 0$ is impossible: indeed, since
		$\beta_n$ tends to a constant, this would require that $q_n \sim
		1/n$, contrary to the condition. Therefore, $\beta_n \to 1/\gamma \Leftrightarrow p_n\sqrt{q_n n} \to \gamma$. In the case where the measurement graph is complete, $q_n=1$, the phase transition with respect to the level of outlier-type corruptions is then exactly at $p_n=1/\sqrt{n}$, which is the threshold that was claimed in the analysis of \cite{singer2011angular} (though only for $\Gr=U(1)$).
	\end{remark}

	\begin{remark}
		Consider also the case where $p_n=1$, that is, we have only additive Gaussian noise, and no outlier-type corruptions. Suppose, moreover, that all of the pairwise measurements are available, that is, $q_n=1$. In that case, we observe a phase transition at noise level $\sigma_n = \sqrt{n}$. This result is well-known, see, for example, the discussion in \cite{perry2016optimality}. In this context, it is also worth to mention the result of \cite{bandeira2017tightness}, who prove that when $\sigma_n \sim n^{1/4}$, the semidefinite relaxation to the maximum likelihood estimator recovers, with high probability, the signal \emph{exactly} in synchronization over $U(1)$. Their numerical experiments suggest that this, in fact, should be true already for noise levels that are almost up to $\sqrt{n}$ (say, $\sigma_n=\frac{\sqrt{n}}{\textrm{poly}\log(n)}$). Note that their type of result is not directly comparable with ours, since we are considering here a \emph{weaker} notion of recoverability (namely, small average squared alignment error). 
	\end{remark}

	\begin{remark}
	\label{rem:phi_stability}
		While the statistic suggested in Theorem \ref{thm:limiting_statistics}
		is indeed exact in the large $n$ limit, we found that it is not
		practical to use when the noise level is close or below the
		recoverability threshold (i.e, when $\gamma \lesssim 1$ in Theorem
		\ref{thm:limiting_mse}). We suspect that the problem here is that of numerical
		stability. Specifically, the analysis in the next sections shows that when $\gamma \le 1$, we have $\eta\to 1$, and so we expect to observe in this case a value of $\eta$ slightly bigger than $1$. However, the derivative of $\hat{\phi}$ (with respect to $\eta$) actually has a $1/\sqrt{\cdot}$ singularity near $\eta=1$, making $\hat{\phi}$ very sensitive to slight deviations in $\eta$.
	\end{remark}

	\paragraph{Outline.}
	This paper now proceeds as follows. 
	In Section \ref{sec:lower} we comment on lower bounds for the noise
	threshold in various groups. 
	In Section \ref{sec:experiments} we report numerical evidence 
	on the finite-$n$ behavior of our
	results. 
	In Section \ref{sec:red} we prove our main results Theorem~\ref{thm:limiting_mse} and Theorem~\ref{thm:limiting_statistics}. The main idea is to show that the measurement matrix $Y$ can be approximately written as a low-rank signal matrix, plus an independent noise-only matrix. This essentially reduces the model at hand to
	the additive-noise variant of
	Johnstone's celebrated Spiked Model \cite{johnstone2001distribution}. We can then use existing results on the limiting eigenvalues and eigenvectors of such models to deduce the limiting MSE of the spectral method. 
	This requires us to prove some results on the limiting spectrum of random block Hermitian matrices with Haar-distributed blocks, which, to the best of our knowledge, do not appear in the literature and could be of independent interest. This is done in Section~\ref{sec:rmt}.
	In Section \ref{sec:proofs} we
	provide proofs to some of the technical claims made before. 
	To make this text self-contained, in the Appendix
	we summarize necessary background from harmonic
	analysis.

%
%

	 \section{Lower bounds for the noise threshold} \label{sec:lower}

	Let us consider, for now, the case where there is no additive noise ($\sigma_n=0$), so that the only type of noise we have are outlier-type corruptions. Theorem \ref{thm:limiting_mse} (see also Remark~\ref{rem:no_additive_noise}) identifies a noise threshold $p_n\sqrt{q_n} = \frac{1}{\sqrt{n}}$ above which (that is, when $p_n\sqrt{q_n}<1/\sqrt{n}$) the estimate returned by the spectral method is completely non-informative. One naturally wonders, then, whether this is the best we can do, among all possible recovery algorithms. Several results from the literature hint that, in some cases, this may well be the case. 

	\paragraph{Finite groups.} 
	Consider the case where $\Gr$ is a finite group of size $L$. In a recent work \cite{perry2016optimality}, it was proved that in the case of a full measurement graph ($q_n=1$) it is impossible to distinguish reliably (with success probability tending to $1$ as $n\to\infty$) between the measurement matrix $Y$ and pure noise whenever 
	\[
	p_n\sqrt{n} \le \sqrt{\frac{2(L-1)\log(L-1)}{L(L-2)}} \,.
	\]
	Moreover, in the region 
	\[
	\sqrt{\frac{4\log L}{L-1}} \le p_n \sqrt{n} \le 1
	\] 
	it is shown that there is an \emph{inefficient} algorithm that can distinguish between $Y$ and pure noise, hinting that there might be a group synchronization algorithm which has a non-trivial MSE even when the spectral method completely fails. These results complement the information-theoretic lower bound given previously in \cite{singer2011angular}. To the best of our knowledge, it is currently unknown whether there exists an \emph{efficient} algorithm that can distinguish between $Y$ and pure noise below the threshold $p_n \le 1/\sqrt{n}$. Nonetheless, these results imply that in the case of synchronization over a finite group, the spectral method is rate-optimal (in the sense that the optimal noise threshold must scale like $p_n \sim 1/\sqrt{n}$).

		\paragraph{Infinite groups.}
In the absence of additive noise, 
	the spectral method is not rate-optimal, at least among inefficient recovery algorithms. In this case, exact recoverability is governed by the \emph{edge-connectivity} of the measurement graph, in a sense which we now describe.
	Let $\Lambda'$ be the graph induced by all the "good" measurements (that is, the measurements that are both available, and which were not replaced by a random group element. Recall that under our model, $\Lambda'$ is an Erd\H{o}s--R\'enyi  graph with edge probability $p_n q_n$. Of course, $\Lambda' \subset \Lambda$). If $\Lambda'$ is \emph{bridgeless} (or $1$-edge connected), meaning that if we remove any single edge the graph remains connected \footnote{Equivalently, any pair of vertices $i$ and $j$ is contained in some simple cycle of $\Lambda'$.}, then the following inefficient procedure recovers $g_1,\ldots,g_n$: we traverse every simple cycle in $\Lambda$. If the group elements along the cycle sum up to the identity - this means that no edge in the cycle was corrupted (that is, the entire cycle is contained in $\Lambda'$), since $g_{ij}$ has zero probability to be any single value. Since $\Lambda'$ is bridgeless, we can identify a connected, noiseless, edge set of $\Lambda$ and use it to reconstruct $g_1,\ldots,g_n$. The threshold for bridgelessness coincides with that for having a Hamiltonian cycle, which is the same as that for having minimum degree $\ge 2$. By the results of \cite{komlos1983limit}, whenever $p_n q_n \ge  (\log(n) + \log\log(n) + \omega(1))/n$, 
	$\Lambda'$ is bridgeless with probability tending to $1$ (when $p_n q_n \le (\log(n)+\log\log(n)+o(1))/n$ that probability tends to $0$, and when $p_nq_n = (\log(n)+\log\log(n)+c)/n $ the probability tends to $e^{-e^{-c}} \in (0,1)$).
	This should be compared with the threshold for connectivity, $p_n q_n =
	\log(n)/n$, which gives us a sharp (to leading order) phase transition between perfect and
	impossible recoverability at $p_n q_n = (1+o(1))\frac{\log(n)}{n}$, which is
	much smaller than the implied noise threshold given in Theorem \ref{thm:limiting_mse}.

	\paragraph{Infinite groups, alternative noise models.}
	The case of an infinite compact group has mostly been investigated in the
	literature under various noise models. We mention the work \cite{bandeira2013cheeger}, which derives worst-case guarantees on the alignment error in $O(d)$ synchronization in terms of the eigenvalues of a certain graph Laplacian, under an adversarial noise model; and \cite{boumal2013cramer} which derives Cramer-Rao -type lower bound for synchronization over $SO(d)$, under a more general noise model that includes both outlier-type corruptions and continuous noise 
	(extending other works which derived such bounds in a less general setting,
	e.g, \cite{chang2006cramer,ash2007relative,howard2010estimation}). In
	particular, the results of \cite{boumal2013cramer} imply that in the absence of additive noise, the rate $p_n\sim
	1/\sqrt{n}$ of the noise threshold is indeed the right one. The analysis of
	\cite{perry2016optimality} for $U(1)$ synchronization imply that under
	Gaussian noise (and no outliers, meaning $p_n=1$, with full measurements, meaning $q_n=1$), the noise threshold $\sigma_n = \sqrt{n}$ is exactly optimal
	(in a suitable sense).

	\section{Finite-$n$ behavior}\label{sec:experiments}

	Before proceeding to prove our main results, we pause to evaluate the
	accuracy of our results in the non-asymptotic regime. 
	We performed numerical simulations to test the validity of our asymptotic predictions (Theorems \ref{thm:limiting_mse} and \ref{thm:limiting_statistics}) in the case of finite sample size $n$, in various regimes of the parameters $q$, $p$ and $\sigma$. In most of the experiments below, we consider a setup where there are only outlier-type corruptions, without additional additive noise (that is, we take $\sigma=0$). 
	This decision is justified in that giving precise asymptotics under outlier-type corruptions is, really, the main technical contribution of this paper. In the case where there is only additive Gaussian noise (and, say, $q=1$), the asymptotic MSE (up to rounding) of the spectral method is already well-known, and follows immediately (without needing to prove anything new!) from the existing results on the extreme eigenvectors in the Spiked Model \footnote{Indeed, in that case the measurement matrix $Y$ \emph{is}, precisely, an independent low-rank (additive) perturbation of a Gaussian random matrix.}. Only in the last experiment of this section, we study a setup which include both outlier-type corruptions and additive Gaussian noise.

	\subsection{Around the recoverability threshold}
	\label{subsect:experiment1}

	We run the spectral method on a random problem instance with a dense measurement graph, $q\in \left\{ 0.5,1 \right\}$. We let the corruption level $p$ scale like $p=\gamma/\sqrt{n}$, with $\gamma$ varying around the theoretical asymptotic threshold $\gamma^{*}=1/\sqrt{q}$ (please see Remark~\ref{rem:no_additive_noise}). We compare the observed MSE proxy (normalized by the dimension of the representation) against the limiting value given in Theorem \ref{thm:limiting_mse} and against the estimate $\hat{\phi}$ suggested in Theorem \ref{thm:limiting_statistics}.
	We do this for several choices of groups: $\Gr = \mathbb{Z}_3, U(2), SO(3)$ represented as rotation subgroups (the representations have dimensions $1,2,3$ respectively). In all the experiments below, we reconstruct from $n=400$ samples. Our results are summarized in Figure \ref{fig:experiment1}.

	\begin{figure}[tbh!]
	       \begin{subfigure}{.5\textwidth}
	            \centering
	            \includegraphics[width=\textwidth]{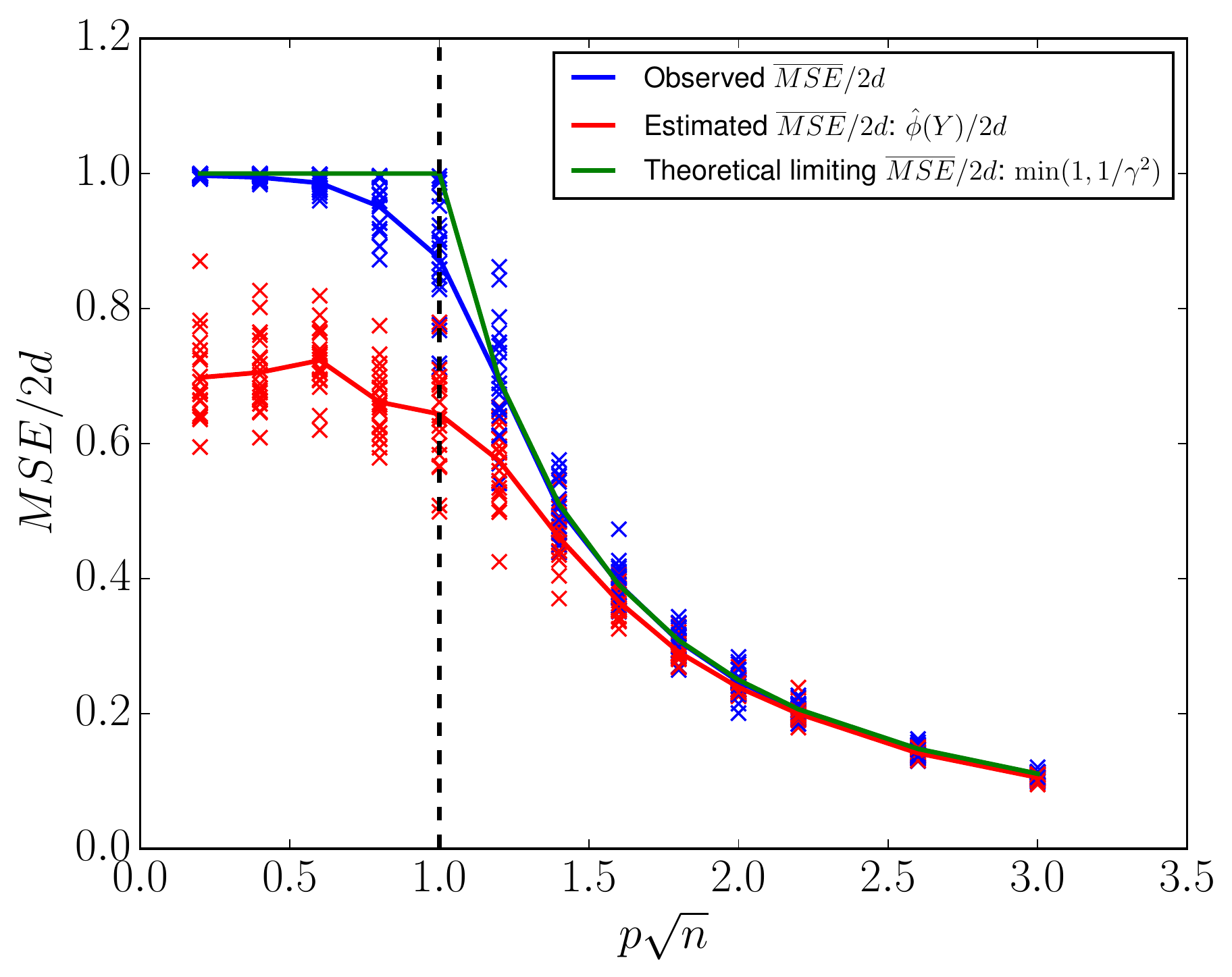}
	      \end{subfigure}
	      \begin{subfigure}{.5\textwidth}
	            \centering
	            \includegraphics[width=\textwidth]{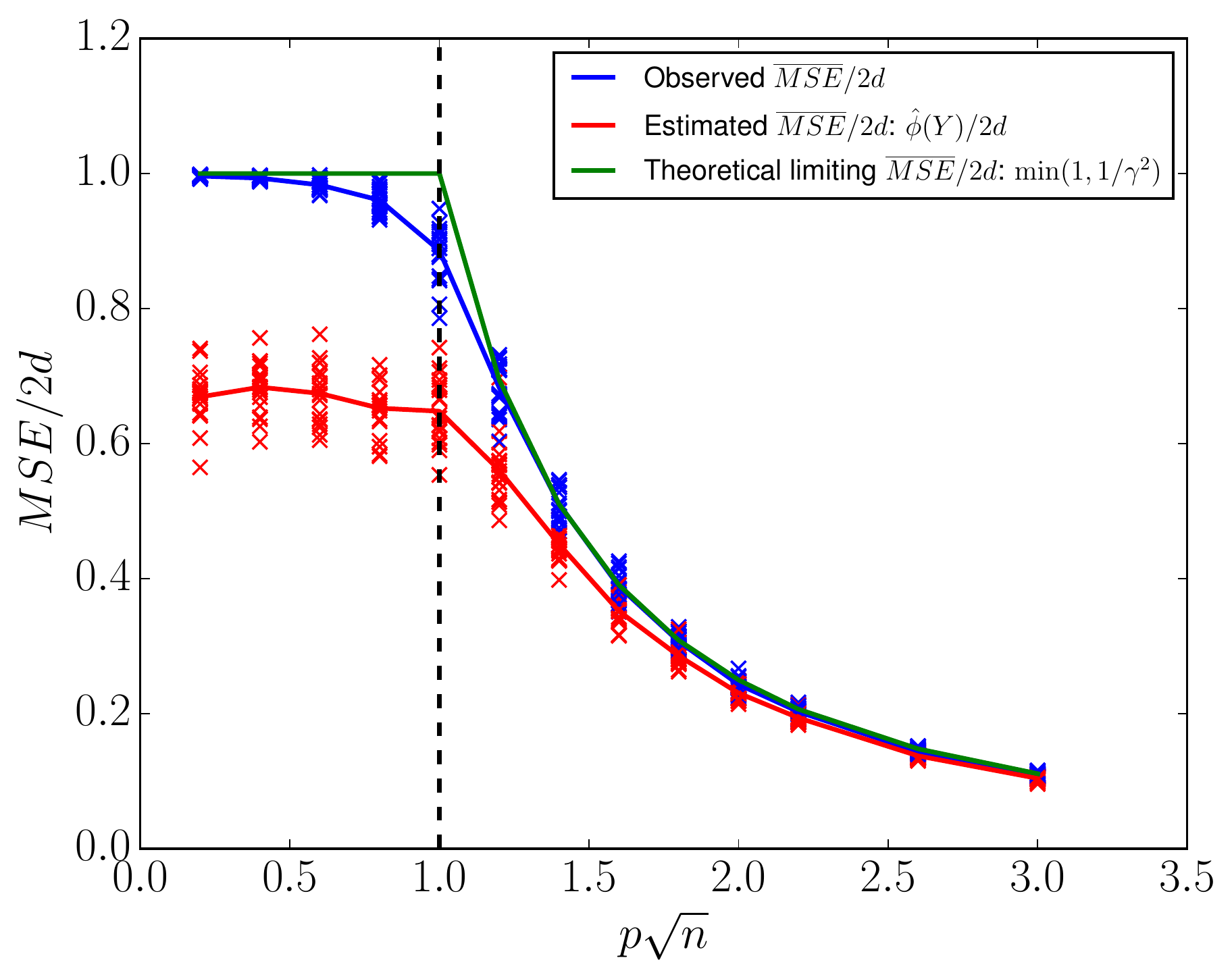}
	      \end{subfigure}
	      \begin{subfigure}{.5\textwidth}
	            \centering
	            \includegraphics[width=\textwidth]{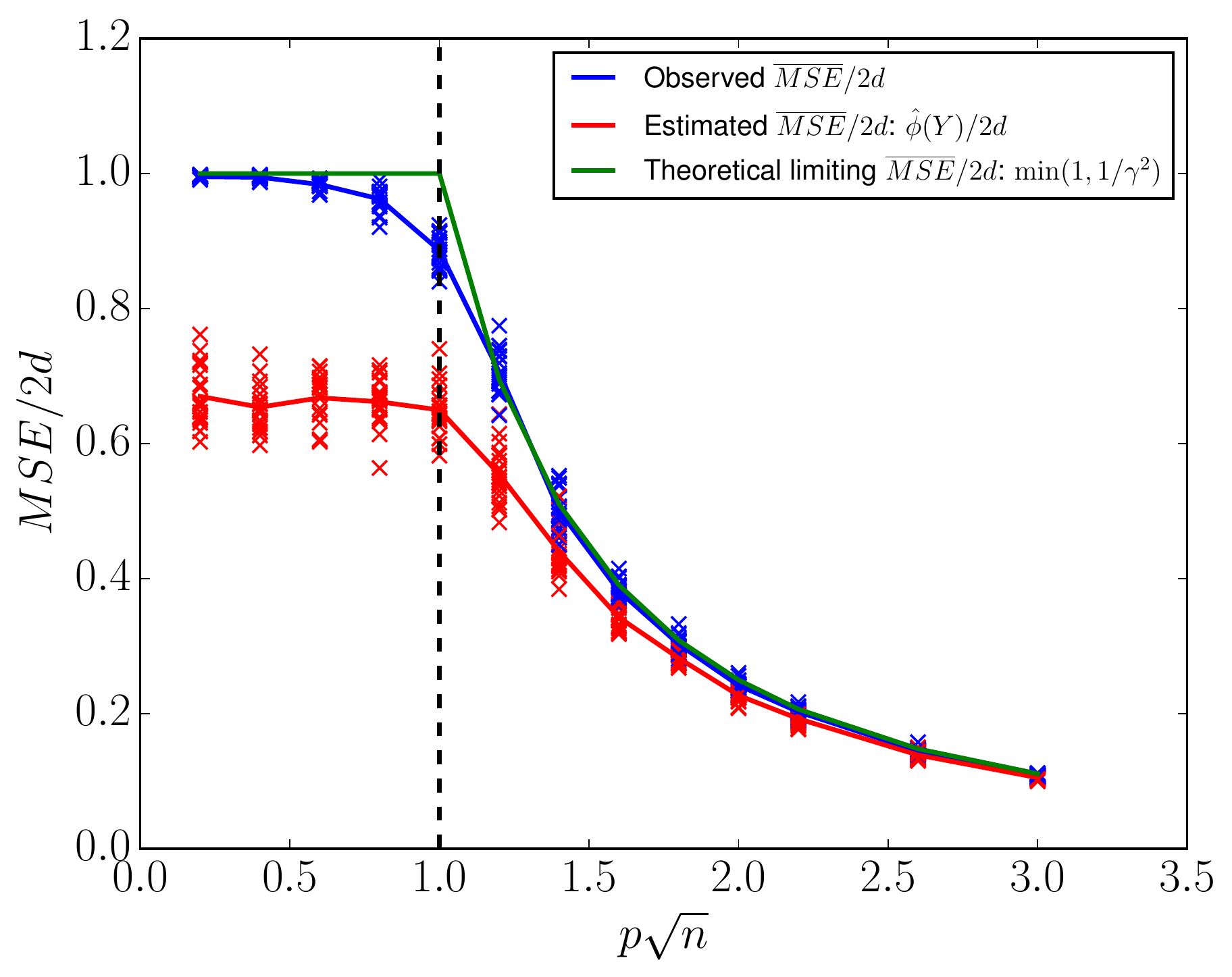}
	      \end{subfigure}
	      \begin{subfigure}{.5\textwidth}
	            \centering
	            \includegraphics[width=\textwidth]{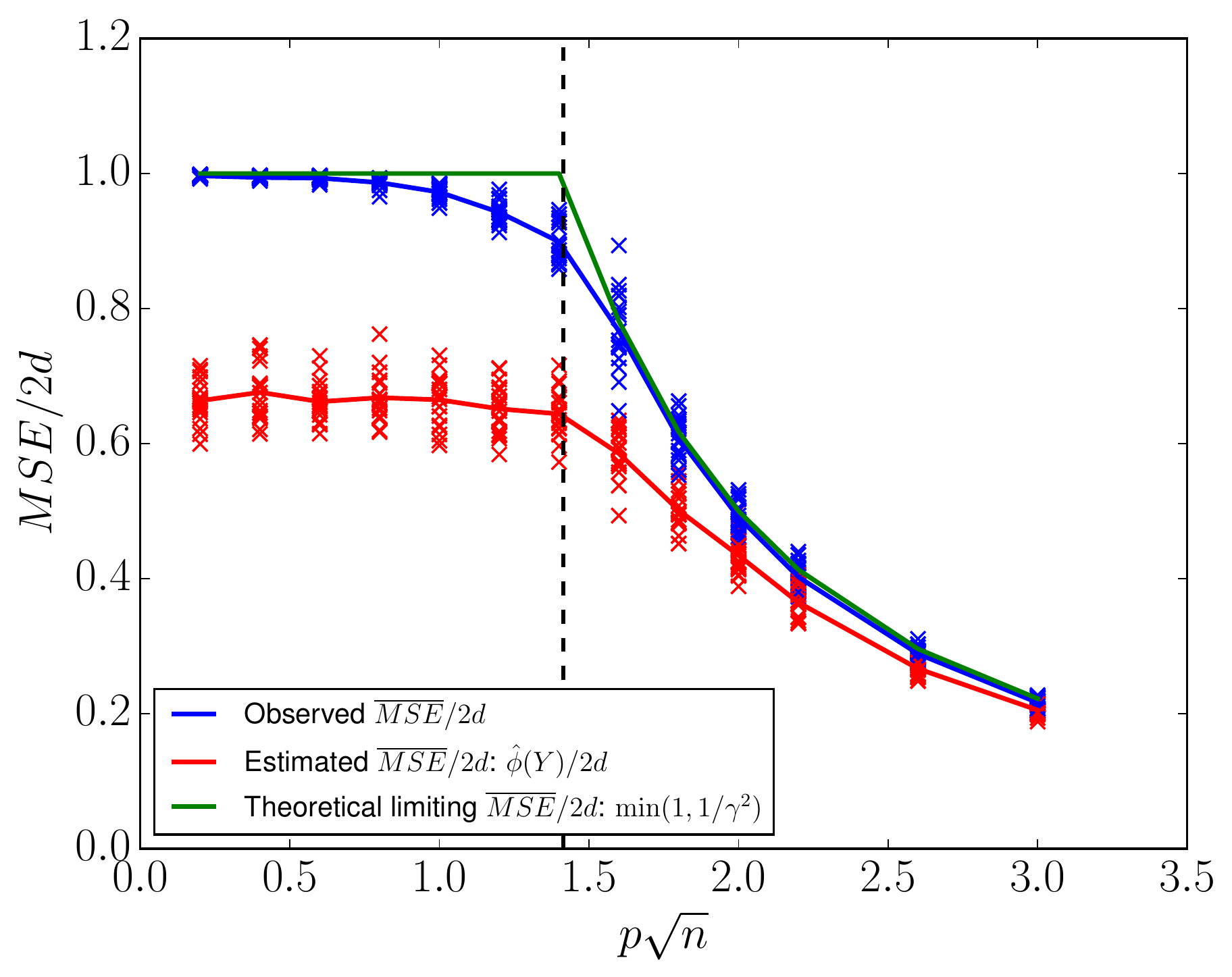}
	      \end{subfigure}
		  \caption{
		  Results of the experiment outlined in Subsection \ref{subsect:experiment1}. Each point on the curve is the average of $T=20$ individual problem instances (the individual results are marked by `x'-s). The dashed vertical line marks the asymptotic threshold $\gamma^{*}=1/\sqrt{q}$. In all plots except for the bottom right, we used $q=1$; there we used $q=0.5$. Top left: $\mathbb{Z}_3$; top right: $U(2)$; bottom: $SO(3)$.}
	      \label{fig:experiment1}
	\end{figure}

	We find that the observed MSE indeed matches the theoretical limiting value quite closely, across all groups. As for $\hat{\phi}$, we see that around the threshold, $\gamma^*$, it predicts the MSE very poorly; as the SNR increases, however, it seems to match the true MSE better and better. See also Remark \ref{rem:phi_stability}, where the numerical stability of $\hat{\phi}$ is discussed; we hypothesize that this is the main reason for the discrepancy between the theory and the observed behavior.

	\subsection{Broad range of corruption levels over a dense measurement graph}
	\label{subsect:experiment2}

	This time, we let the $p$ run over a broad range of values, with the logarithmic scaling $p=n^{-e}$ as $e=0,0.1,\ldots,1$. Here $q=1$ and $n=400$, making $p$ range from $0.0025$ to $1$. We used $\Gr=SO(3)$ throughout this experiment. Our results are summarized in Figure \ref{fig:experiment2}.

	\begin{figure}[tbh!]
	       \begin{subfigure}{.7\textwidth}
	            \centering
	            \includegraphics[width=\textwidth]{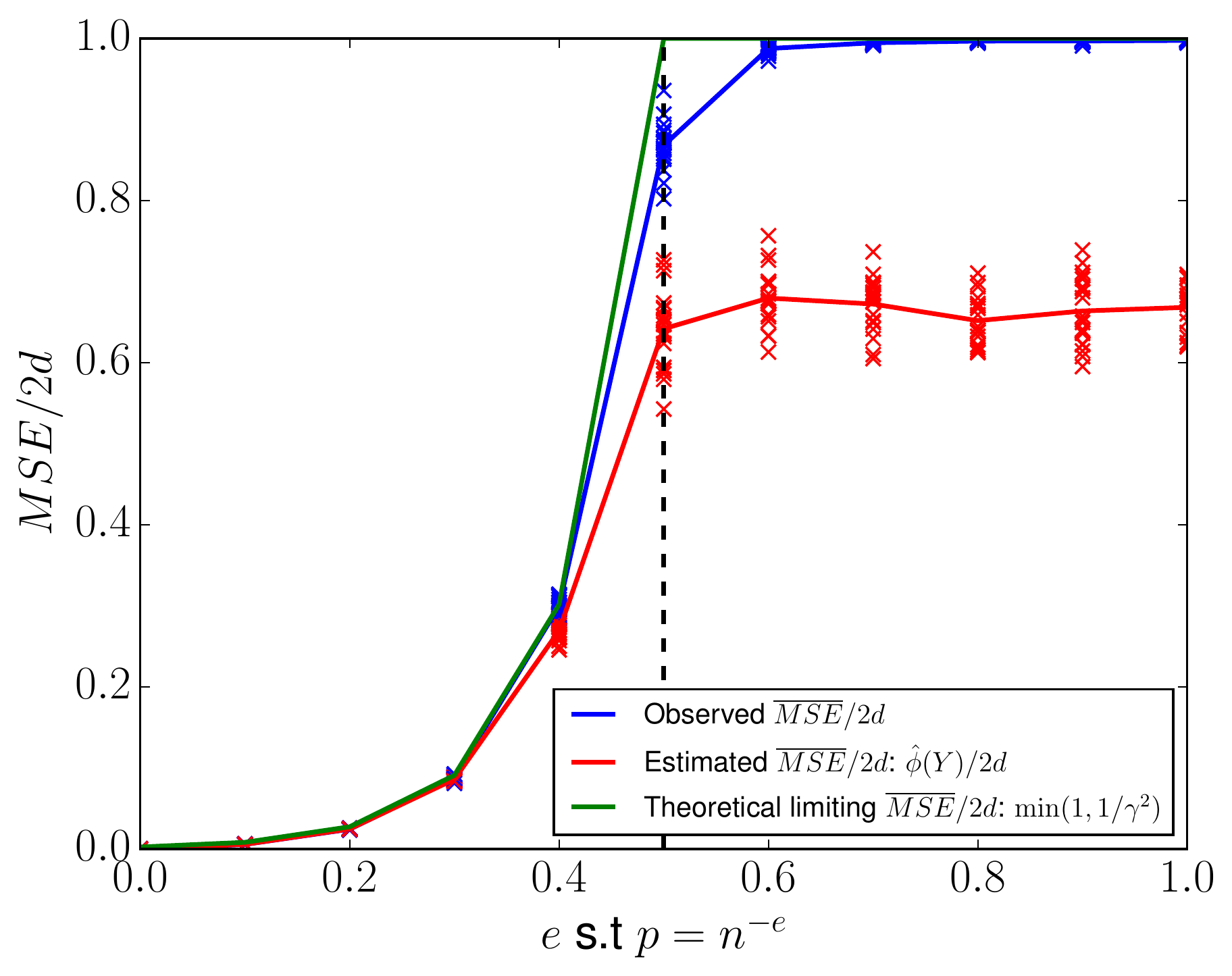}
	      \end{subfigure}
	      \caption{Results of the experiment outlined in Subsection \ref{subsect:experiment2}. Each point on the curve is the average of $T=20$ individual problem instances (the individual results are marked by `x'-s). The dashed vertical line marks the asymptotic threshold $e^{*}=1/2$. Here $\gamma = p\sqrt{qn}$.}
	      \label{fig:experiment2}
	\end{figure}

	We find that the expression of Theorem \ref{thm:limiting_mse} predicts quite correctly the MSE even in the case where $p$ is quite large relative to $1/\sqrt{qn}$ (informally, ``$\gamma \sim \sqrt{n}$'' in the notation of Theorem \ref{thm:limiting_mse}), at least in the case where the measurement graph is complete (here $q=1$).

	\subsection{Sparser measurement graph}
	\label{subsect:experiment3}

	This time, we keep the product $\sqrt{n} \times p_n \sqrt{q_n} = 2 $ fixed and vary the sparsity level of the measurement graph in a logarithmic scale, $q_n = n^{-e}$, so that $p_n = 2n^{-1/2+e/2}$. In that case, $p_n q_n \sim n^{-1/2-e/2}$, so that as long as $e<1$, we are still operating under the conditions of Theorem \ref{thm:limiting_mse}. 	
	We use $n=400$, so that $e$ varies from $0$ to $1-\frac{log(4)}{log(n)} \approx 0.76$, to keep the constraint $p\le 1$. We used $\Gr=SO(3)$ throughout this experiment. Our results are summarized in Figure \ref{fig:experiment3}.

	\begin{figure}[tbh!]
	       \begin{subfigure}{.7\textwidth}
	            \centering
	            \includegraphics[width=\textwidth]{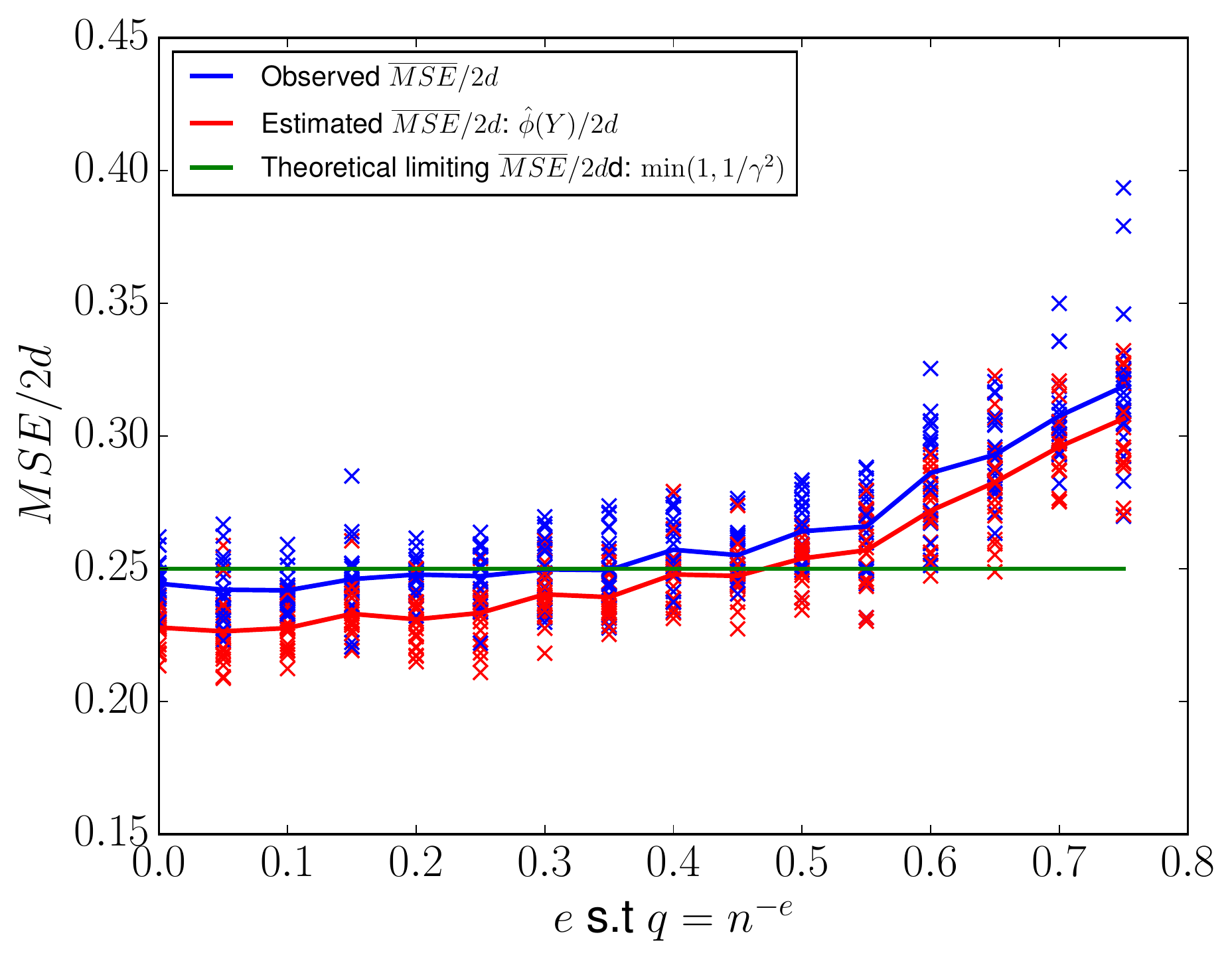}
	      \end{subfigure}
	      \caption{Results of the experiment outlined in Subsection \ref{subsect:experiment3}. Each point on the curve is the average of $T=20$ individual problem instances (the individual results are marked by `x'-s). 
      	}
	      \label{fig:experiment3}
	\end{figure}

	We find that as the measurement graph becomes sparser ($e$ increases, with $n$ being fixed), 
	the discrepancy between the predicted asymptotic MSE from Theorem \ref{thm:limiting_mse} and the observed MSE becomes larger.

\subsection{Convergence to the limit}
\label{subsect:experiment4}

We next want to give evidence for the convergence of the MSE proxy $\overline{MSE}$ and statistic $\hat{\phi(Y)}$ to their limiting value in accordance with the results of Theorem \ref{thm:limiting_mse}. To that end, we estimate the expected squared deviations of these quantities from their limit value,
\[
	\E \left( \overline{MSE}-\lim\overline{MSE} \right)^2,\quad \E \left( \hat{\phi}(Y)-\lim\overline{MSE} \right)^2 
\]
for different values of $n$. In this experiment, we use throughout $\Gr=SO(3)$, $q=1$, $p_n = 5/\sqrt{n}$, and had $n$ range from $50$ to $1200$. Our results are summarized in Figure \ref{fig:experiment4}.
	
\begin{figure}[tbh!]
       \begin{subfigure}{.7\textwidth}
            \centering
            \includegraphics[width=\textwidth]{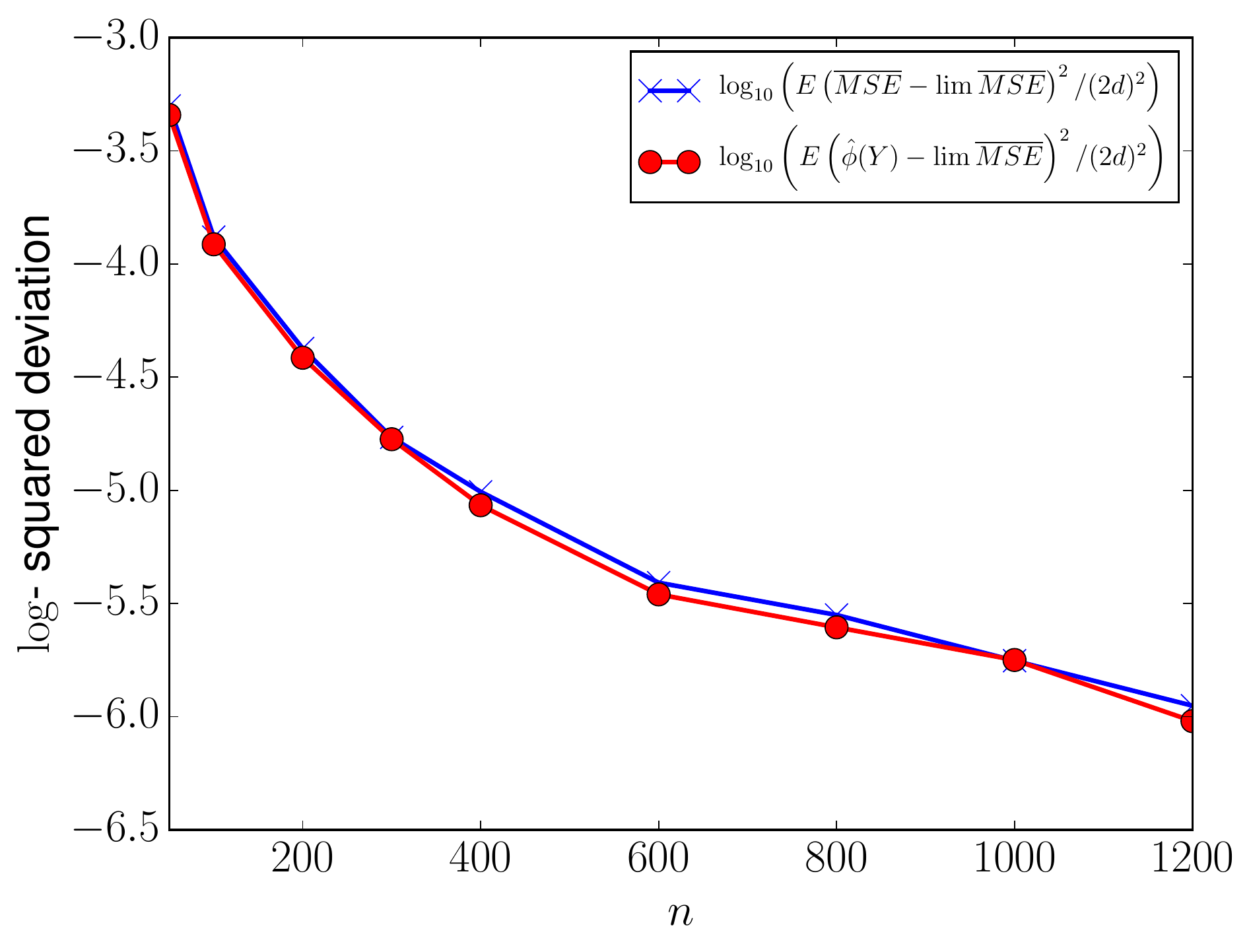}
      \end{subfigure}
      \caption{Results of the experiment outlined in Subsection \ref{subsect:experiment4}. Each point on the curve is the average of the observed squared deviation across $T=20$ Monte-Carlo trials. }
      \label{fig:experiment4}
\end{figure}

We find that indeed the squared deviations decrease as $n$ increases, in affirmation with the theory.

\subsection{How much do we lose by rounding?}
\label{subsect:experiment5}

Up to this point, we only investigated the MSE proxy $\overline{MSE}(X,\tilde{X})$ in itself. We would now like to see if this quantity indeed tells us something practical about the actual problem we set out to solve, namely, about the error $MSE(X,\hat{X})$ after rounding. In the experiment outlined below, we used optimal blockwise rounding, which for the following groups is given explicitly:
\begin{enumerate}
	\item $\Gr = \mathbb{Z}_2$ represented as $\left\{ \pm1 \right\}$. Optimal rounding is given by
	\[
	\textrm{round}(X_i) = \text{sgn}(X_i) \in \left\{ \pm 1 \right\} \,.
	\]
	\item $\Gr = O(3)$ represented as rotation matrices. Optimal rounding is given by
	\[
	\textrm{round}(X_i) = UV^{T}
	\]
	where $X_i = U\Sigma V^T$ is the SVD. Indeed, if $O\in O(3)$,
	\[
	\norm{O-X_i}_F^2 = \norm{X}_F^2 + d - 2\tr \left( O^T X_i \right) \,,
	\]
	where $\tr (O^T X_i) \le \tr \Sigma$, with equality when $O=UV^T$ 
	\footnote{ Recall that by duality for matrix norms,
	\[
		\tr (X^T Y) \le \norm{X} \norm{Y}_{*}
	\]
	where $\norm{\cdot}$ is the operator norm and $\norm{\cdot}_*$ is the
Nuclear norm (sum of singular values). }.
\end{enumerate}

We run the spectral method on a random problem instance with a dense measurement graph, $q\in \left\{ 0.5,1 \right\}$. We let the corruption scale like $p=\gamma/\sqrt{n}$, with $\gamma$ varying around the theoretical asymptotic threshold $\gamma^{*}=1/\sqrt{q}$. We compare the observed and asymptotic MSE proxy, $\overline{MSE}$, with the true MSE obtained after block-wise rounding of the measurement eigenvectors. We do this for the two groups $\Gr=\mathbb{Z}_2,O(3)$, where in each experiment we attempt recovery from $n=400$ samples. Our results are summarized in Figure \ref{fig:experiment5}.

	\begin{figure}[tbh!]
		  	       \begin{subfigure}{.5\textwidth}
	            \centering
	            \includegraphics[width=\textwidth]{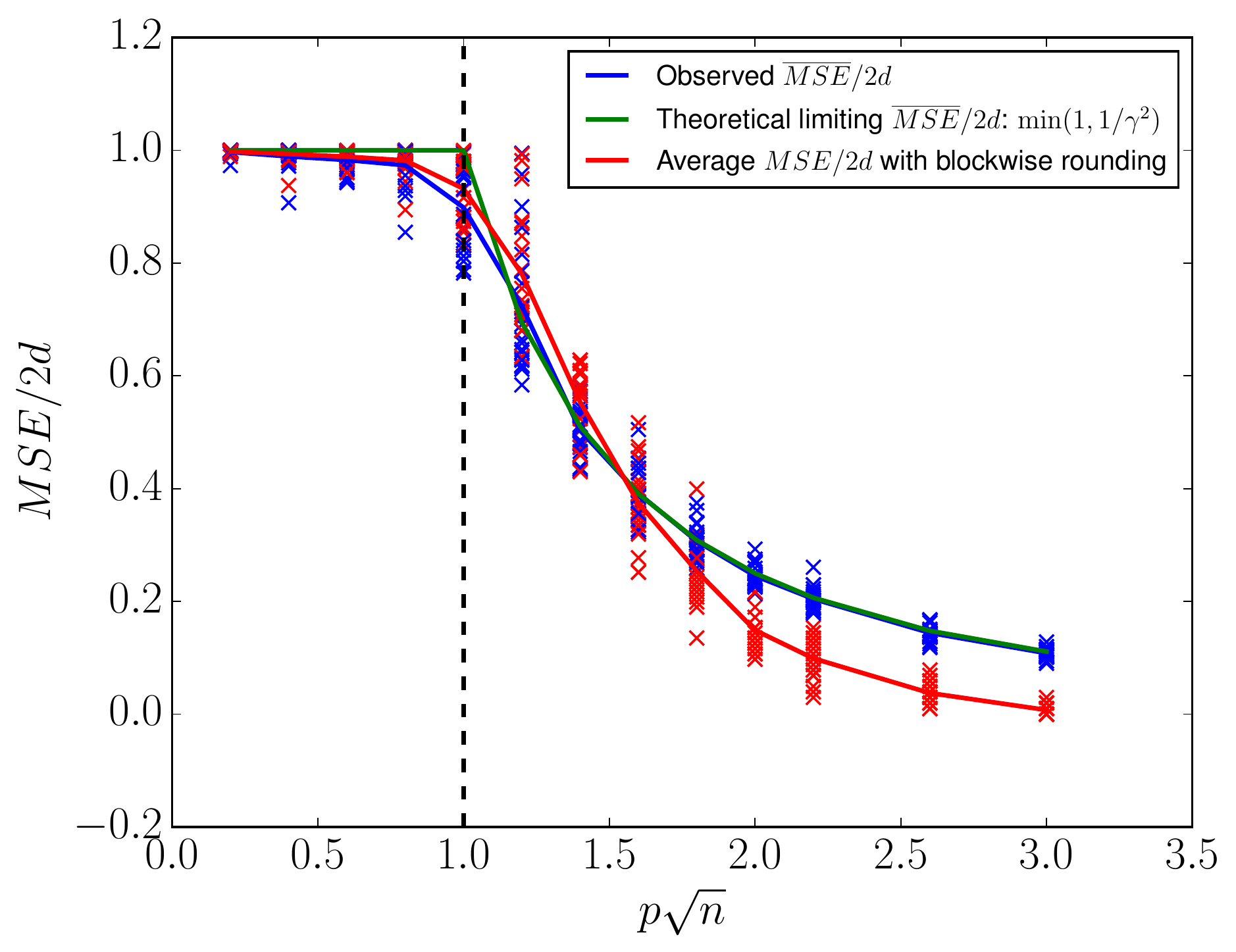}
	      \end{subfigure}
	      \begin{subfigure}{.5\textwidth}
	            \centering
	            \includegraphics[width=\textwidth]{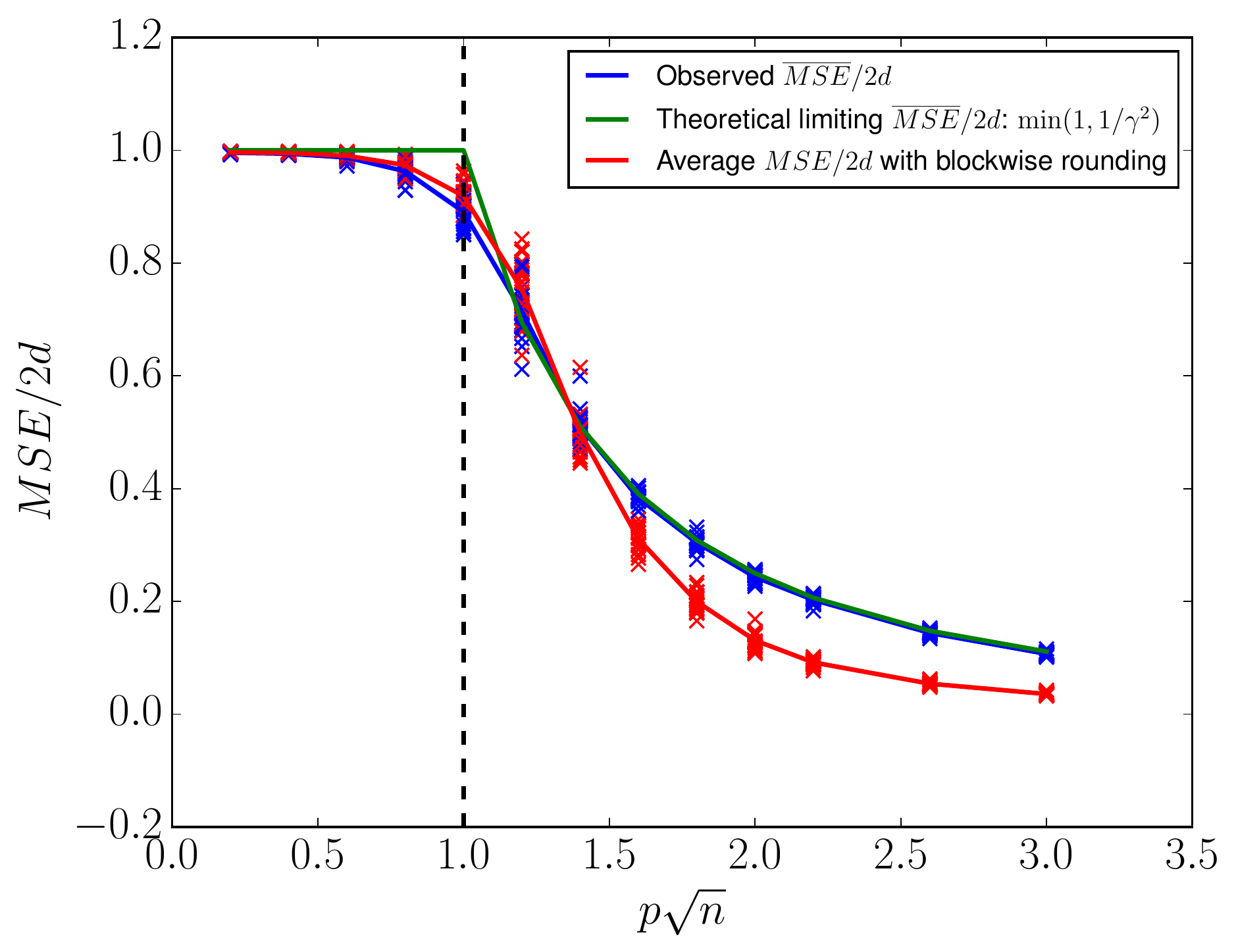}
	      \end{subfigure}
	      \begin{subfigure}{.5\textwidth}
	            \centering
	            \includegraphics[width=\textwidth]{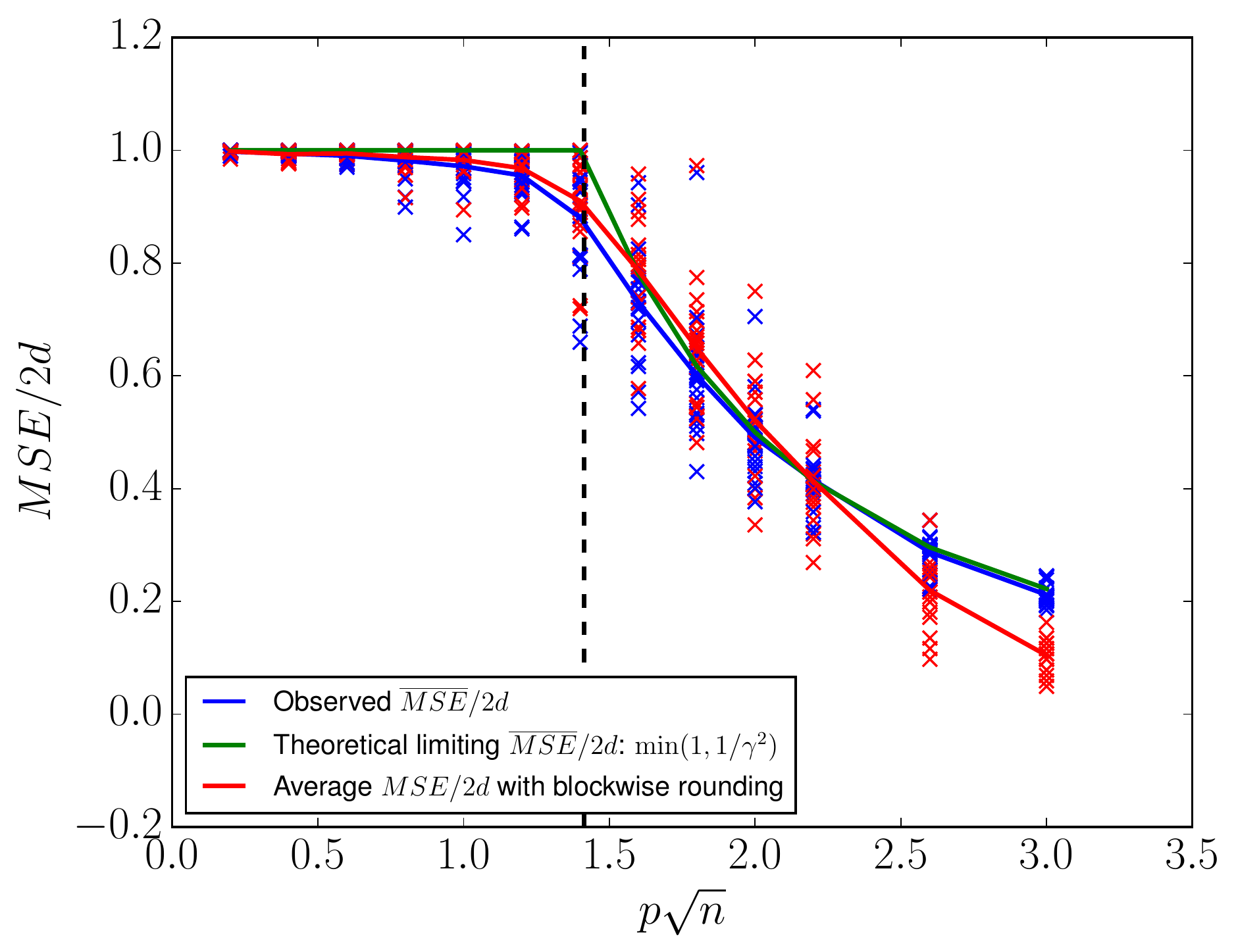}
	      \end{subfigure}
	      \begin{subfigure}{.5\textwidth}
	            \centering
	            \includegraphics[width=\textwidth]{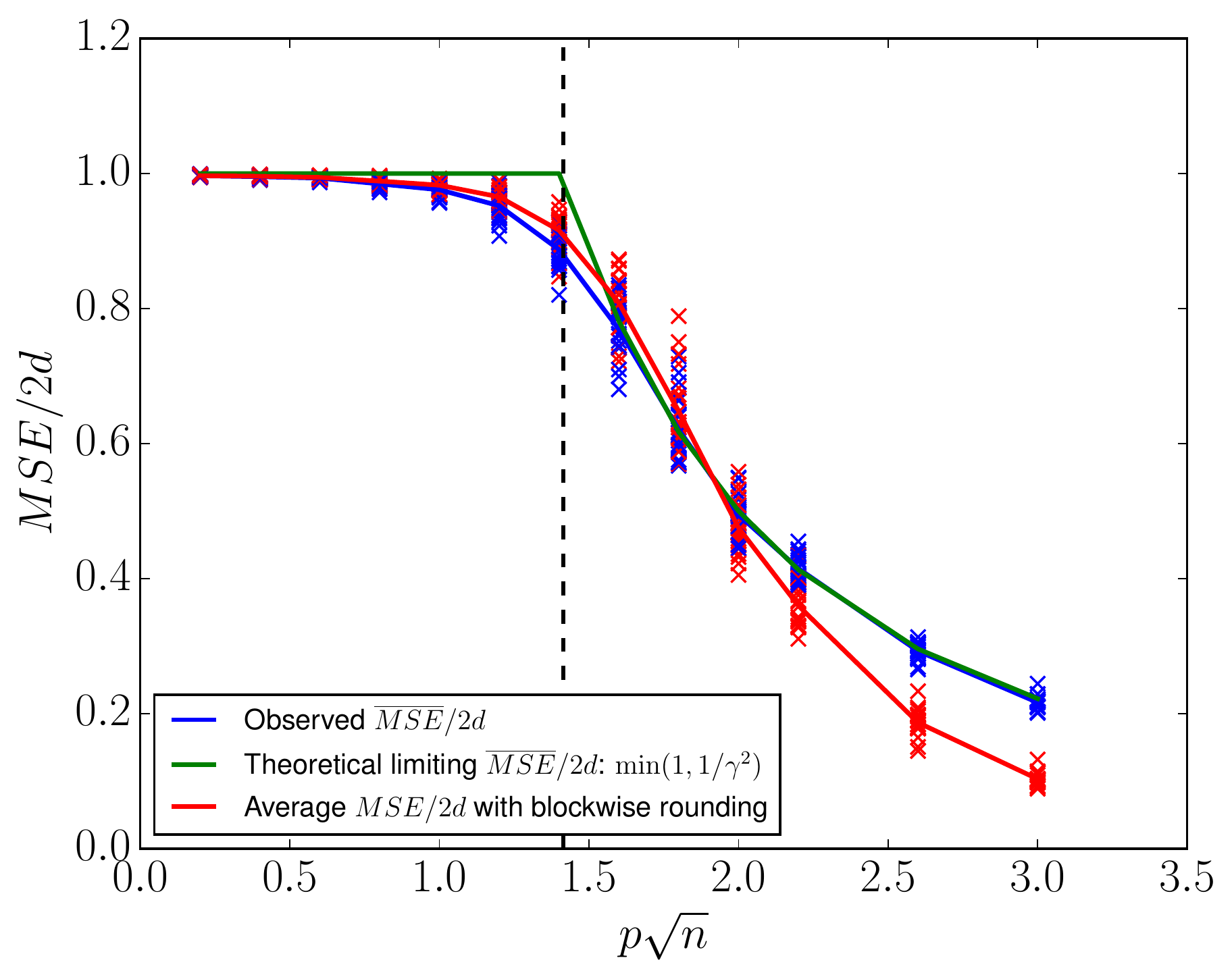}
	      \end{subfigure}

		  \caption{
		  Results of the experiment outlined in Subsection
	  \ref{subsect:experiment5}. Each point on the curve is the average of
  $T=20$ individual problem instances (the individual results are marked by
  `x'-s). The dashed vertical line marks the asymptotic threshold
  $\gamma^{*}=1/\sqrt{q}$. Left: $\Gr=\mathbb{Z}_2$; right: $\Gr=O(3)$. Top:
  $q=1.0$; bottom: $q=0.5$.}
	      \label{fig:experiment5}
	\end{figure}

We find that the rounded MSE displays a noise-sensitivity phase transition at roughly the same location as the phase transition for $\overline{MSE}$. Indeed, at high noise (low $\gamma$) $\overline{MSE}$ appears to be a reliable proxy for MSE; at low noise, it appears to give consistently pessimistic estimates for the rounded MSE.

\subsection{Adding additive noise}
\label{subsect:experiment6}

We now consider the case where we have both outlier-type corruptions, and also (real valued, Gaussian) additive noise. In all the experiments in this section we take $q=1$, $n=400$ and $\Gr=SO(3)$. 

Recall that $1/\gamma \approx \beta = \frac{\sqrt{1-p+\sigma^2}}{p\sqrt{n}}$, so that solving for $\sigma$ in terms of $p$ and $\gamma$, we have
\[
\sigma = \sqrt{\frac{p^2 n}{\gamma^2}+p-1}\,.
\]
For three choices of $p$: (1) $p=3/\sqrt{n}=0.15$; (2) $p=2/n^{1/4}\approx0.45$; and (3) $p=1$, we run random recovery experiments so that $\sigma$ is chosen according to the expression above, as to give $\gamma\in \{0.5,1,\ldots,3\}$. Our results are summarize in Figure~\ref{fig:experiment6}. In all the cases considered, we observe good agreement with the asymptotic predictions of Theorems~\ref{thm:limiting_mse} and \ref{thm:limiting_statistics}.

\begin{figure}[tbh!]
	\begin{subfigure}{.5\textwidth}
		\centering
		\includegraphics[width=\textwidth]{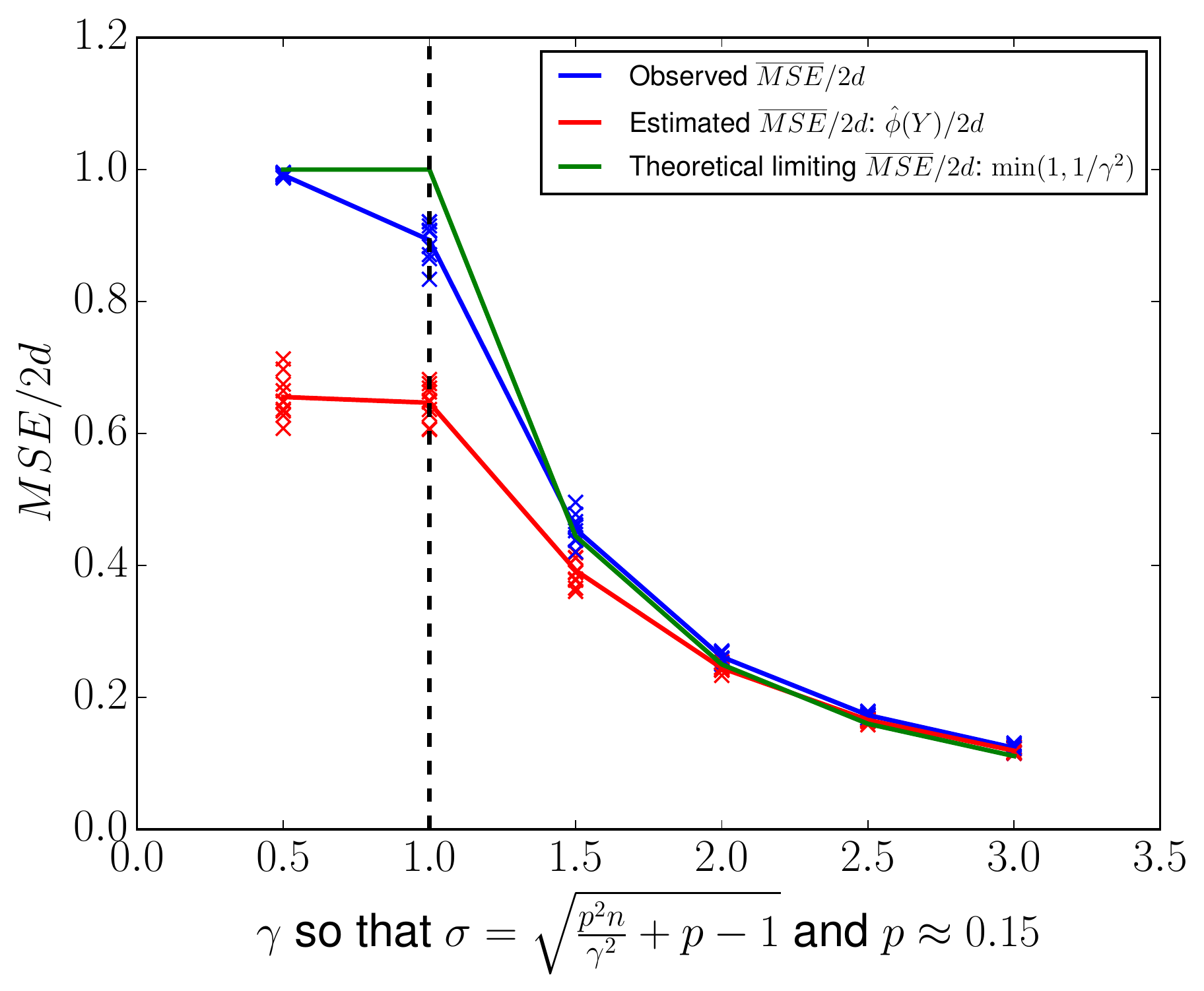}
	\end{subfigure}
	\begin{subfigure}{.5\textwidth}
		\centering
		\includegraphics[width=\textwidth]{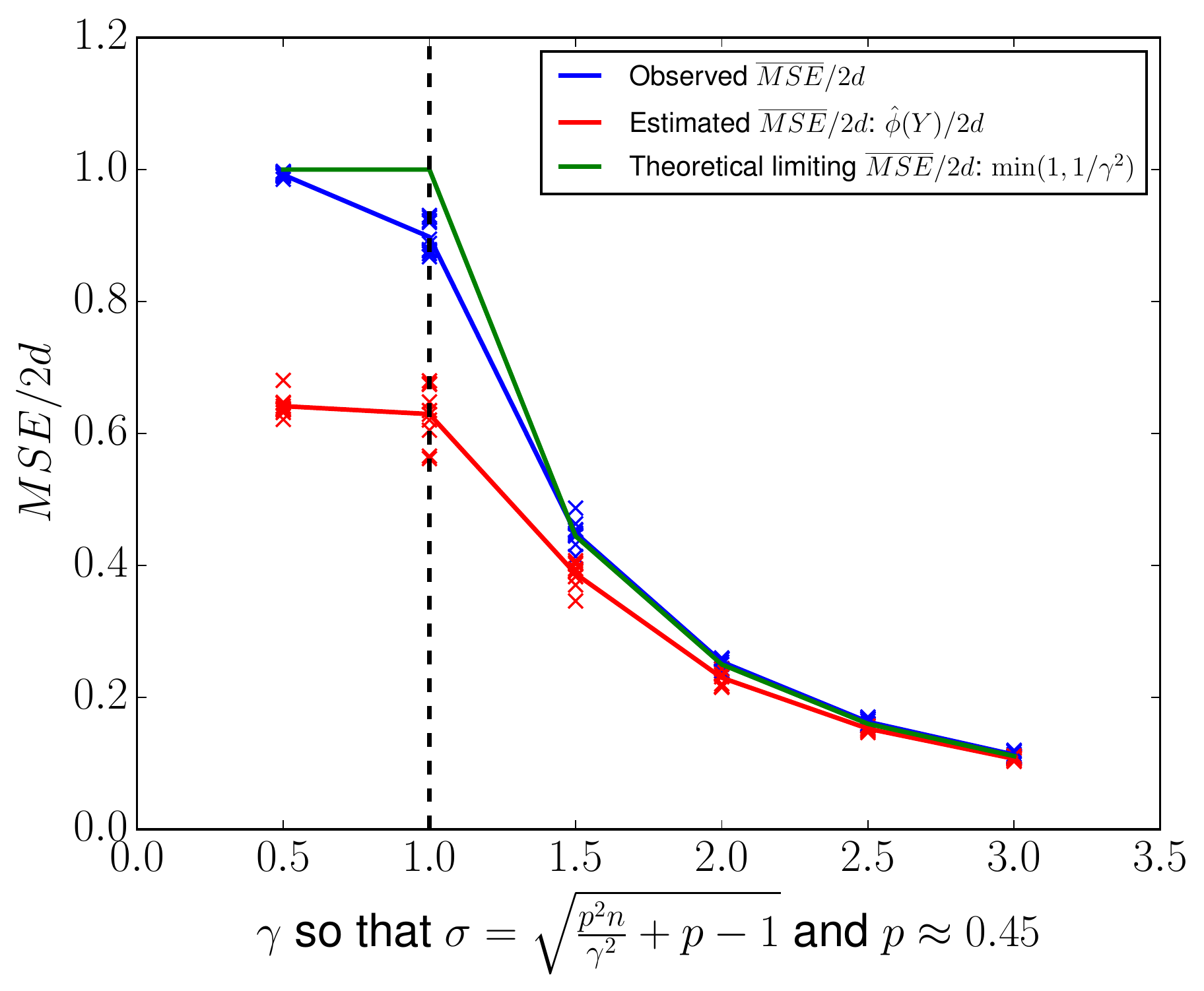}
	\end{subfigure}
	\begin{subfigure}{.5\textwidth}
		\centering
		\includegraphics[width=\textwidth]{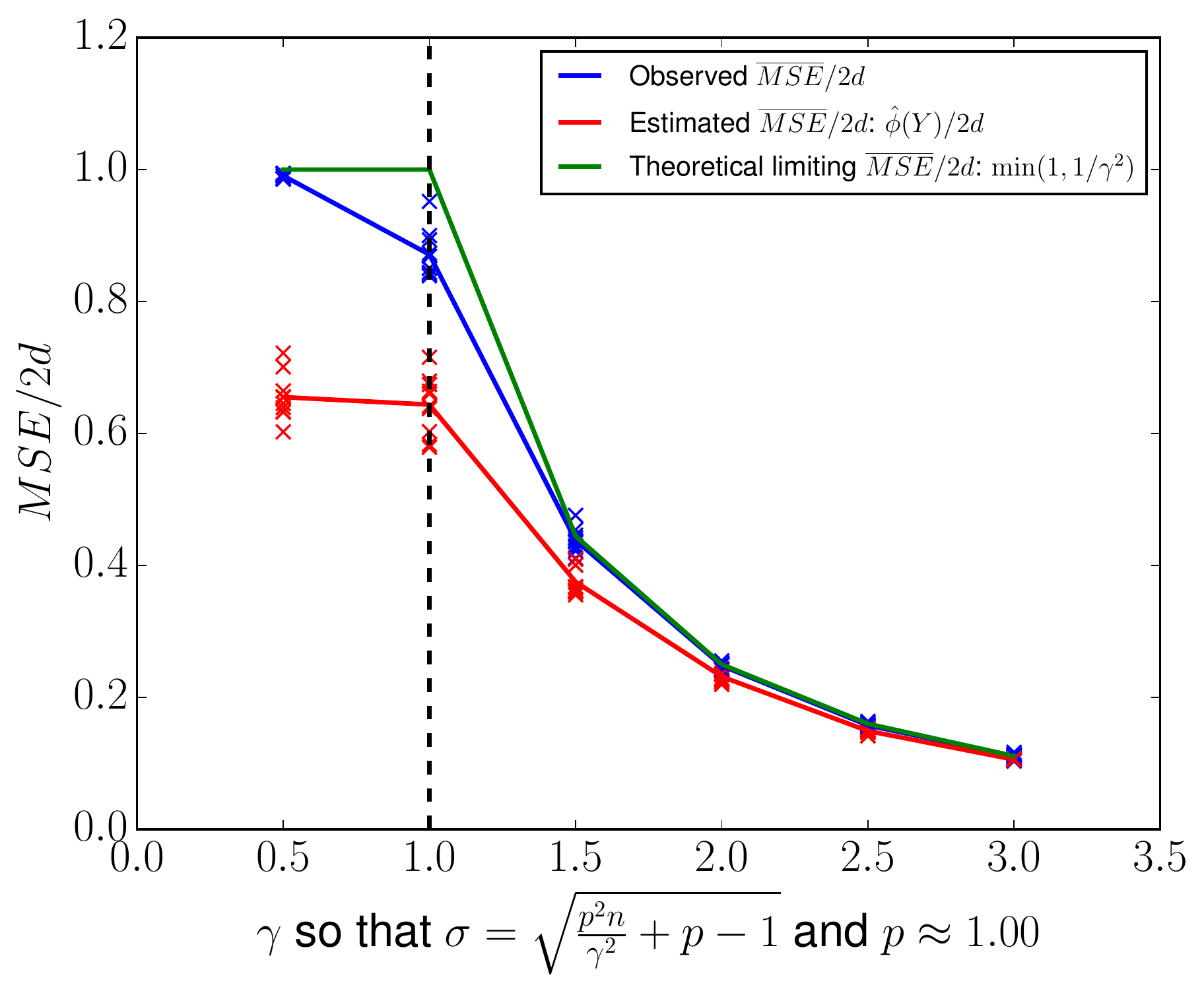}
	\end{subfigure}
	\caption{
		Results of the experiment outlined in Subsection \ref{subsect:experiment6}. Each point on the curve is the average of $T=10$ individual problem instances (the individual results are marked by `x'-s). The dashed vertical line marks $\gamma=1$, the location of the theoretically predicted noise threshold. }
	\label{fig:experiment6}
\end{figure}

\section{Proof of the main results} \label{sec:red}

We now turn to prove our main results Theorems \ref{thm:limiting_mse} and \ref{thm:limiting_statistics}. We first introduce some notation, which will shall use throughout the rest of the analysis.

Let $\Pi$ be an $nd\times nd$ Hermitian block matrix, such that every $d\times d$ block above the diagonal is a random element of the form $\pi(g_{ij})$, with $g_{ij}\sim \textrm{Haar}$. Let $\Ec$ be a Hermitian block matrix (of the same dimensions) such that every block above the diagonal is a matrix of the form $\epsilon_{ij}$, which is a $d\times d$ i.i.d Gaussian (either real or complex) matrix, with mean $0$ and variance (in the complex case - absolute second moment) $1$. The block diagonal of both $\Pi$ and $\Ec$ is zero. Let $\Delta$ be a symmetric block matrix, such that every block above the diagonal is either all-ones with probability $p_n$, or all-zeros with probability $1-p_n$. Let $E$ be a matrix of the same kind as $\Delta$, with $p_n$ replace by $q_n$. The block diagonal of $\Delta$ and $E$ is all-ones. The matrices $\Pi$, $\Ec$, $\Delta$ an $E$ are all independent of one another, and of the signal $X$. 

The none-zero blocks of $E$ will correspond to the measurements that are available to us, that is, $(i,j)\in\Lambda$ if and only if the $(i,j)$-th block of $E$ is all-ones. Similarly, the blocks of $\Delta$ will correspond to those links that were not replaced by complete noise (that is, we observe a possible noisy version of the true group difference). 

We can now write our measurement matrix $Y$ in the form
\begin{equation}
	Y = E \odot \Delta \odot (XX^*) + E \odot \left( \overline{\Delta}\odot \Pi + \frac{\sigma_n}{\sqrt{d}}\Ec \right) \,,
\end{equation}
where $\odot$ denotes entry-wise (Hadamard) product, $\overline{\Delta}=\one\one^*-\Delta$ (the "complement" of $\Delta$), and $\one$ is an $nd$-dimensional all-ones vector. The first summand above will be thought of as the "signal", and second one will be thought of as the "noise". Observe that at this point, these two random matrices are dependent on one another (through $E$ and $\Delta$). Our first task is to show that as $n\to\infty$, this dependence asymptotically decouples. 

We start by rescaling $Y$, and consider for now the matrix
\begin{equation}
	Y_1 := \frac{1}{n q_n p_n}Y = \frac{1}{n q_np_n} E \odot \Delta \odot (XX^*) + \frac{1}{n q_np_n}E \odot \left( \overline{\Delta}\odot \Pi + \frac{\sigma_n}{\sqrt{d}}\Ec \right) \,.
\end{equation}
Consider also the matrix 
\begin{equation}
	Y_2 = \frac{1}{n}XX^* + \frac{1}{n q_np_n}E \odot \left( \overline{\Delta}\odot \Pi + \frac{\sigma_n}{\sqrt{d}}\Ec \right) \,,
\end{equation}
where we note that 
\[
\E_{E,\Delta} \left[ \frac{1}{n q_n p_n} E \odot \Delta \odot (XX^*) \right] = \frac{1}{n}XX^* 
\]
(here we take the expectation only with respect to $E$ and $\Delta$). Our first result states that asymptotically, $Y_1$ and $Y_2$ become arbitrarily close to one another in operator norm. This would imply, in a manner which will be made precise later, that in order to study the top $d$ eigenvalues and eigenvectors of $Y_1$, it suffices to study those of $Y_2$.

\begin{proposition}
	\label{prop:signal_noise_decouple}
	There is some universal numerical constant $c_0>1$ (that doesn't depend on the representation $\pi$ or the group we use) such that if $p_nq_n = \omega(\log^{c_0}(n)/n)$. Then 
	\[
	\norm{Y_1-Y_2} \aslim 0 \,,
	\]
	as $n\to\infty$. 
\end{proposition}
The proof of Proposition \ref{prop:signal_noise_decouple} is technical, and deferred to Subsection~\ref{subsection:signal_noise_decouple_proof}.

Our present goal, then, is to study the top eigenvectors of the matrix $Y_2$.
We denote
\begin{equation}
	\Wc = \frac{1}{\tau_n \sqrt{n}} E \odot \left( \overline{\Delta}\odot \Pi + \frac{\sigma_n}{\sqrt{d}}\Ec \right) \,,
\end{equation}
where 
\begin{equation}
	\tau_n := \sqrt{{q_n}\left(1-p_n+\sigma_n^2\right)} \,.
\end{equation}
We also denote
\begin{equation}
\begin{split}
	\beta_n 
	&=  \frac{\sqrt{n}\tau_n}{n q_n p_n} \\
	&= \frac{1}{p_n\sqrt{q_n n}} \sqrt{1-p_n+\sigma_n^2} \,,
\end{split}
\end{equation}
which will be thought of as a noise-to-signal ratio. Note that the normalization was chosen so that the off-(block-)diagonal elements of $\Wc$ all have mean $0$ and second absolute moment $E\abs{\Wc_{lk}^2}=1/(nd)$. We now write the matrix $Y_2$ as 
\begin{equation}
	Y_2 = \frac{1}{n}XX^* + \beta_n \Wc\,.
\end{equation}

\subsection{Analysis of the highly noisy case, using the theory of low-rank perturbations}

Observe that the matrix $\frac{1}{n}XX^*$ has rank $d$, and all of its none-zero eigenvalues are exactly $1$ (and an orthonormal basis for the corresponding eigenspace is given by the columns of $X/\sqrt{n}$). This low-rank component is now, of course, independent of $\beta_n \Wc$. The matrix $Y_2$ can be thought of as a low-rank perturbation of a random matrix. Models of this form are very common is modern statistics and engineering (see, for example, the books \cite{bai2010spectral,couillet2011random}), and many properties about their limiting extreme eigenvalues and eigenvectors can be calculated, in terms of properties of the unperturbed random matrix $\Wc$.  

In order to leverage existing results about models of this kind, we need to study the limiting properties of the spectrum of $\Wc$. In particular, we will need two components: (1) to calculate the limiting distribution of the eigenvalues (if such exists); (2) to calculate the almost-sure limits (if those exist) of the extreme (largest and smallest) eigenvalues. We prove the following two results, which could also be of independent interest:

\begin{theorem}[Semicircle law for Hermitian block matrices with Haar-distributed entries]
	\label{thm:bulk}
	Suppose that either one of the following holds:
	\begin{enumerate}
		\item (no outliers) $p_n=1$, and $q_n = \omega \left(1/n\right)$. 
		\item (with outliers, but enough noise) $\tau_n^2 = q_n(1-p_n+\sigma_n^2) = \omega\left(1/n\right)$ and $q_n = \omega\left(1/n\right)$. In particular, $\sigma_n = \Omega(1)$ is enough to satisfy the first requirement. When $\sigma_n=0$, the first condition reduces to $q_n(1-p_n)=\omega(1/n)$ - that is, that there are sufficiently many outliers.
	\end{enumerate}
	Then as $n\to\infty$, the empirical spectral distribution (ESD) of $\Wc$ converges weakly almost surely to the Wigner semicircle law. That is, denoting the eigenvalues of $\Wc$ by $\lambda_1 \ge \ldots \ge \lambda_{nd}$, we have
	\[
	\frac{1}{nd} \sum_{i=1}^{nd} \delta_{\lambda_i} \to f_{sc}(\lambda)d\lambda \,,
	\]
	almost surely (where by convergence we mean convergence of measures in the weak sense). Recall that the Wigner semicircle law is given by the density 
	\begin{equation}
	\label{eq:wiger-semicircle}
	f_{sc}(\lambda) = \frac{1}{2\pi} \sqrt{4-\lambda^2} \cdot
	\mathbf{1}_{ {\abs{\lambda}\le 2} }
	\end{equation}
	with respect to Lebesgue measure.
\end{theorem} 

\begin{theorem}[The extreme eigenvalues converge to the bulk edge]
	\label{thm:edge}
	There is a universal numeric constant $c_1>1$ such that the following holds. Suppose that either of the following holds:
	\begin{enumerate}
		\item (no outliers) $p_n=1$ and $q_n=\omega(\log^{c_1}(n)/n)$.
		\item (with outliers, but enough noise) $p_n<1$, $\tau_n = \omega(\log^{c_1}(n)/n)$ and $q_n = \omega(\log^{c_1}(n)/n)$.
	\end{enumerate}
Then
	\begin{equation}
		\lambda_{\max}(\Wc) \to 2\,\quad \lambda_{\min}(\Wc) \to -2 \,,
	\end{equation}
	almost surely as $n\to\infty$. Note that $\pm 2$ are the edges of the support of the Wigner semicircle law, Eq. \ref{eq:wiger-semicircle}.
\end{theorem}
Theorems \ref{thm:bulk} and \ref{thm:edge} are the main technical contributions of this paper. Their proofs are presented in Section~\ref{sec:rmt}. The conditions in Theorems \ref{thm:bulk} and \ref{thm:edge} may seem somewhat unintuitive at first sight. Their essence is this: in general, we disallow measurement graphs that are too sparse, specifically, that the average degree $nq_n$ is  too small \footnote{Note that $q_n n=\log(n)$ is the (scale of the) threshold for connectivity.}. Second, when there are no outliers ($p_n=1$), we adopt a scaling such that the non-zero entries of $\Wc$ are simply i.i.d Gaussians with unit variance (regardless of $\sigma_n$), so there shouldn't be any further problems. However, when there are outliers ($p_n<1$) we also have to account for the matrix $\Pi$, and ensure that our scaling for $1/\tau_n \sqrt{n}$ indeed makes the entries very small (rather than blow them up, instead).

Equipped with Theorems \ref{thm:bulk} and \ref{thm:edge}, we shall now leverage the results of \cite{benaych2011eigenvalues} to characterize the top $d$ eigenvalues and eigenvectors of $Y_2$, working under the assumptions of Theorems \ref{thm:bulk} and \ref{thm:edge}. First, let us identify the parameter regime in which we can obtain non-trivial asymptotics: (1) We note that when the noise-to-signal parameter $\beta_n \to 0$, we have that $\norm{Y_2-\frac{1}{n}XX^*} \to 0$. Hence the $d$ top eigenvalues of $Y_2$ converge to $1$, and the eigespace asymptotically aligns with the signal eigenspace. It will be obvious that in this case, we have that the MSE proxy $\overline{MSE} \to 0$. (2) On the other hand, when $\beta_n \to \infty$, the effective signal is completely swamped by the noise and the top eigenvectors of $Y_2$ essentially become independent of the signal. In that case, the top eigenvalues all behave like $\beta_n(2+o(1)))$. It will be obvious that in this case, $\overline{MSE} \to 2d$. 

By the discussion above, we now concentrate on the case where $\beta_n \to 1/\gamma$, where $\gamma\in(0,\infty)$ is a constant. Denoting 
\begin{equation}
	Y_3 = \frac{1}{n}XX^* + \gamma^{-1}\Wc \,,
\end{equation}
Theorem \ref{thm:edge} ensures us that $\norm{Y_2-Y_3} \to 0$. We would now like to apply the results of \cite{benaych2011eigenvalues}. Note that their stated result requires that either the signal or the noise matrix has a unitarily invariant (orthognally invariant in the real case) distribution. Upon a closer inspection of the proof, we see that it suffices to verify a certain, simple "incoherence" condition between the signal principal components and the noise matrix, which does hold under our model. In order to keep the narrative flow, we defer the statement and proof of this condition (which is otherwise not particularly informative) to Subsection~\ref{subsection:BG-conditions}.

Finally, by \cite{benaych2011eigenvalues} we have that for every $i=1,\ldots,d$, 
\begin{equation}
	\lambda_i(Y_3) \aslim \begin{cases}
	\frac{1}{\gamma} \left(\gamma + \frac{1}{\gamma}\right) &\quad \text{ when } \gamma>1 \\
		\frac{2}{\gamma} &\quad \text{ when } \gamma \le 1 
	\end{cases}\,.
\end{equation}
Denoting by $u_1,\ldots,u_d$ the corresponding eigenvectors of $Y_3$, their projection onto the signal eigenspace satisfies 
\begin{equation}
	\norm{\Pc_{X}(u_i)}^2 \aslim \begin{cases}
	1 - \frac{1}{\gamma^2} &\quad \text{ when } \gamma>1 \\
	0 &\quad \text{ when } \gamma \le 1 
	\end{cases}\,,
\end{equation}
where $\Pc_X$ is the projection onto the column span of $X$, that is, $\Pc_X(u) = \frac{1}{n}XX^*u$. Moreover, for any fixed number $j\ge 1$, we have that 
\begin{equation}
	\lambda_{d+j}(Y_3) \aslim 2/\gamma \,.
\end{equation}

Let us now focus on the case where $\gamma>1$. Let $\tilde{X}$ be an $nd\times d$ matrix whose columns are the top $d$ eigenvectors of $Y_1$ (hence of $Y$), and $U$ be the corresponding matrix for the top eigenvectors of $Y_3$. Since $\norm{Y_1-Y_3}\aslim 0$, as well as $\lambda_{d}(Y_3)-\lambda_{d+1}(Y_3)=\Omega(1)$, we may use the Davis-Kahan theorem (see, for example, Theorem 2 in \cite{yu2014useful}) to conclude that
\[
\norm{UU^*-\tilde{X}\tilde{X}^*}_F \aslim 0 \,.
\] 
Hence,
\begin{align*}
\overline{MSE}(X,\tilde{X}) 
&= \norm{\frac{1}{n}XX^*-\tilde{X}\tilde{X}^*}_F^2 \\
&\approx \norm{\frac{1}{n}XX^*-\tilde{U}\tilde{U}^*}_F^2 \\
&=	2d - 2\sum_{i=1}^{d} \norm{\Pc_X (\tilde{u_i})}^2 \aslim \begin{cases}
	\frac{2d}{\gamma^2}, &\text{ if } \gamma > 1\\
	2d, &\text{ otherwise }\,,
	\end{cases} \,,
\end{align*}
as claimed in Theorem \ref{thm:limiting_mse}. While we may not use the Davis-Kahan theorem when $\gamma\le 1$, it is reasonable to expect that indeed the MSE proxy does not improve at lower $\gamma$. Our numerical experiments in Section~\ref{sec:experiments} suggest that this is indeed the case. Also observe that using the formulas above, we may estimate the signal-to-noise ratio $\gamma$, and hence also the MSE proxy. We can compute the statistic,
\begin{equation}
\begin{split}
\eta 
&= \frac{\lambda_{1}(Y)}{\lambda_{d+1}(Y)} \\
&= \frac{\lambda_{1}(Y_1)}{\lambda_{d+1}(Y_1)} \\
&\approx \frac{\lambda_{}(Y_3)}{\lambda_{d+1}(Y_3)} \\
&\approx \begin{cases}
\frac{1}{2}\left( \gamma + \frac{1}{\gamma} \right) &\quad \text{ when } \gamma > 1 \\
1 &\quad \text{ when } \gamma \le  1 \\
\end{cases} \,.
\end{split} 
\end{equation}
We may solve this equation for $\gamma$, to obtain
\[
\gamma \approx \eta + \sqrt{\eta^2-1} \,.
\]
Plugging this into the expression for the MSE, we obtain
\[
\overline{MSE}(X,\tilde{X}) \approx \frac{2d}{\left( \eta + \sqrt{\eta^2-1} \right)^2}\,,
\]
and notice that when $\gamma<1$, using $\eta=1$ indeed gives us the right value for $\overline{MSE}(X,\tilde{X})$. Also observe that (perhaps unsurprisingly) the statistic on the right should converge to the correct quantity even in the regime where this exact analysis doesn't hold, i.e, $\beta_n \to  0,\infty$ (which could be thought of, at least formally, as the limits $\gamma \to \infty,0$ of the analysis here).  

To conclude, we have proved Theorems \ref{thm:limiting_mse} and \ref{thm:limiting_statistics} under the additional assumptions of Theorems \ref{thm:bulk} and \ref{thm:edge}. We shall take $c$, the constant in the statement of Theorem \ref{thm:limiting_mse}, to be $c=\max(c_0,c_1)$, where $c_0$ and $c_1$ are the constants from Proposition \ref{prop:signal_noise_decouple} and Theorem \ref{thm:edge} respectively.

\subsection{The case of very low noise}

There is something deeply unsatisfying with our analysis up to this point: in order for our proofs to work, we needed to assume that the noise level is sufficiently high, in the sense that if $p_n>0$, it must be that $n\tau_n^2 = \omega(\log^c(n)/n)$. In particular, when $\sigma_n=0$ (no additive noise), this means that $q_n(1-p_n)$ has to be large, hence $p_n$ small! It is quite clear, however, that performance must improve, rather than degrade, as we increase $p_n$ (that is, \emph{decrease} the number of outlier-type corruptions). 

Recall that $\beta_n = \frac{\sqrt{n}\tau_n}{np_n q_n}$. Working under the assumption of Theorem \ref{thm:limiting_mse}, we have that $np_nq_n = \omega(\log^c(n))$. The case not covered by the previous analysis is when $\lim\inf n\tau_n^2/\log^c(n) < \infty$, or equivalently that $\lim\inf \sqrt{n}\tau_n/\log^{c/2}(n) < \infty$. We find that in this case, necessarily $\beta_n \to 0$, with much room to spare: we in fact have the much stronger condition $(\sqrt{np_nq_n})\beta_n \to 0$ as well! In the case where $1/(\sqrt{n}\tau_n)$ is large, the scaling we chose for $\Wc$ appears to be the wrong one. Indeed, when $\Wc$ is very sparse, a result of the kind of Theorems \ref{thm:bulk} and \ref{thm:edge} is too much to hope for. Fortunately, it is also not necessary for our purpose: we only need to show that $\norm{\beta_n \Wc} \to 0$. 

The following proposition, which is proved in Subsection~\ref{subsection:small-beta}, completes the proof of our main results.
\begin{proposition}
	\label{prop:small-beta}
	There exists a (universal) choice $c\ge \max(c_0,c_1)$ such that if $np_nq_n = \omega(\log^c(n))$, and also $\sqrt{np_nq_n}\beta_n = o(1)$, then we have that $\norm{\beta_n \Wc} \aslim 0$. 
\end{proposition}
Using this $c$ in the assumption of Theorems \ref{thm:limiting_mse} and \ref{thm:limiting_statistics} would ensure that they hold even in the case where it is not true that $n\tau_n^2 = \omega(\log^c(n)/n)$, that is, where the analysis of the previous subsection failed.

\section{Spectrum of the pure-noise matrix: proof of Theorems \ref{thm:bulk} and \ref{thm:edge}} \label{sec:rmt}
 
\subsection{Proof of Theorem \ref{thm:bulk}}

The result will follow from the main Theorem of \cite{girko1996matrix} (note that while the result there is stated for symmetric real matrices, the same proof works for the Hermitian case). We need to verify the following three conditions (their precise implications will be mentioned and explained in detail right after):

Let $\Wc_{ij} = \Wc_{ij}^{(n)}$ ($i,j=1,\ldots,n$) be the $d\times d$ blocks of the Hermitian block matrix $\Wc$.
\begin{enumerate}
  \item {\bf First moment.} The matrix $\Wc$ needs to be centered, that is $\E(\Wc)=0$. This is indeed the case here.
  \item {\bf Second moment.} The blocks need to satisfy that
	\[
		\sup_{n} \max_{i=1,\ldots,n} \sum_{j=1}^{n} \E \norm{\Wc_{ij}}_F^2 < \infty \,.
	\]
	In our case, we normalized $\Wc$ so that each off-diagonal block has exactly $\E \norm{\Wc_{ij}}_F^2 = d^2 \cdot (1/nd) = d/n$, so this condition is certainly satisfied. 
      \item {\bf Lindeberg condition.} The blocks need to satisfy a Lindeberg-type condition: for every fixed $\alpha>0$, 
	\[
	\lim_{n\to\infty} \max_{i=1,\ldots,n} \sum_{j=1}^{n} \E \left[ \norm{\Wc_{ij}}_F^2 \I \left\{ \norm{\Wc_{ij}}_F >\alpha \right\} \right] = 0 \,.
	\]
	In the case where $\sigma_n=0$ (case (3) in the theorem statement), this is very easy to verify: we have that with probability $q_n(1-p_n)$,  $\norm{\Wc_{ij}}_F^2 = d/(q_n(1-p_n) n)$ (which tends to $0$ as $n$ grows, by assumption) and with probability $1-q_n(1-p_n)$ that $\Wc_{ij}=0$; in any case, for large enough $n$, $\norm{\Wc_{ij}}_F \le \alpha$ with probability $1$, so that each summand above becomes identically zero. As for the case $\sigma_n>0$ with $p_n<1$ (case (2) in the theorem statement), let us use the crude bound
	\begin{align*}
		\norm{\Wc_{ij}}_F 
		&\le \frac{1}{\tau_n\sqrt{n}}\left(\norm{\pi(g_{ij})}_F + \frac{\sigma_n}{\sqrt{d}}\norm{\epsilon_{ij}}_F\right) \\
		&= \frac{1}{\tau_n\sqrt{n}}\left(\sqrt{d} + \frac{\sigma_n}{\sqrt{d}}\norm{\epsilon_{ij}}_F\right) \,,
	\end{align*}
	where $\pi(g_{ij})$ and $\epsilon_{ij}$ are the blocks of $\Pi$ and $\Ec$ respectively. Recall that $\tau_n \sqrt{n} \to\infty$. For all large enough $n$ such that $\frac{\sqrt{d}}{\tau_n \sqrt{n}} < \alpha/2$, we have that for every $t\ge \alpha$,
	\begin{align*}
		\Pr \left( \norm{\Wc_{ij}}_F > t \right)
		&\le q_n \Pr \left( \frac{1}{\tau_n\sqrt{n}}\left(\sqrt{d} + \frac{\sigma_n}{\sqrt{d}}\norm{\epsilon_{ij}}_F\right) > t \right) \\
		&\le q_n \Pr \left( \frac{1}{\tau_n\sqrt{n}}\cdot \frac{\sigma_n}{\sqrt{d}}\norm{\epsilon_{ij}}_F > t/2  \right) \\
		&= q_n\Pr \left( \norm{\epsilon_{ij}}_F^2 > \frac{d\tau_n^2 {n}}{4\sigma_n^2}\cdot t^2  \right) \,.
	\end{align*}
	Notice that when the additive noise $\epsilon_{ij}$ is real, $\norm{\epsilon_{ij}}_F^2$ is simply a sum of $d^2$ squared  independent $N(0,1)$ variables \footnote{That is, it has a $\chi$-squared distribution with $d^2$ degrees of freedom.} (in the complex case there are $2d^2$ squared $N(0,1/2)$ variables). Using the general inequality
	\[
	\Pr\left(\sum_{i=1}^D g_i \ge t\right) \le \sum_{i=1}^D \Pr\left( g_i > t/D\right) \,,
	\]
	we can further bound the probability above as
	\[
	\Pr \left( \norm{\Wc_{ij}}_F > t \right) \le O(1) \cdot d^2 q_n e^{-O\left(\frac{\tau_n^2 {n}}{d\sigma_n^2}\cdot t^2\right)} \,.
	\]
	Now,
	\begin{align*}
		\E \left[ \norm{\Wc_{ij}}_F^2 \I \left\{ \norm{\Wc_{ij}}_F >\alpha \right\} \right]
		&= \int_{0}^{\infty} \Pr \left( \norm{\Wc_{ij}}_F^2 \I \left\{ \norm{\Wc_{ij}}_F >\alpha \right\} \ge t \right)dt \\
		&= \alpha^2 \Pr \left(\norm{\Wc_{ij}}_F^2 >\alpha^2\right) + \int_{\alpha^2}^{\infty} \Pr \left(\norm{\Wc_{ij}}_F^2 >t\right) dt \\
		&\le O(1)\cdot d^2 q_n e^{-\Omega\left(\frac{\tau_n^2 {n}}{d\sigma_n^2}\cdot t^2\right)} \left( \alpha^2 + \sqrt{\frac{1}{\frac{\tau_n^2 {n}}{d\sigma_n^2}}} \right) \,.
	\end{align*}
	Noting moreover that $\tau_n^2/\sigma_n^2 = q_n(1-p_n+\sigma_n^2)/\sigma_n^2 \ge  q_n$, we may bound
	\begin{align*}
		\sum_{j=1}^{n} \E \left[ \norm{\Wc_{ij}}_F^2 \I \left\{ \norm{\Wc_{ij}}_F >\alpha \right\} \right] 
		\le O(1) \cdot d^2 nq_n e^{-\Omega\left(\frac{q_n {n}}{d}\cdot \alpha^2\right)} \left( \alpha^2 + \sqrt{\frac{1}{\frac{q_n {n}}{d}}} \right)
	\end{align*}
	which tends to $0$ whenever $q_n n \to \infty$. As for the case $\sigma_n>0$ and $p_n=1$ (case (3) in the theorem statement), the matrix $\Wc$ is simply an i.i.d Gaussian matrix with random erasures. The proof is very similar to the previous case, where now we don't have to ensure that the contribution from $\Pi$ is asymptotically vanishing (that is why we needed to make sure that $\tau_n \sqrt{n}\to \infty$ before). 
\end{enumerate}

Having satisfied the conditions above, the main result of \cite{girko1996matrix} tells us now that  the following holds:
Denote by $F_{\mu_n}$ the CDF (cumulative distribution function) of the empirical spectral distribution of $\Wc$. Then for almost every point $x$, we have that $F_{\mu_n}(x)-F_n(x) \to 0$ almost surely, where $F_n$ are a sequence of CDFs whose Stieltjes transform satisfies
\[
	\int (x-z)^{-1} dF_n(x) = \frac{1}{d} \tr C(z)\, \quad \Im(z)\ne 0 \,,
\]
where $C(z)$ is the unique $d\times d$ matrix analytic function on $\C\setminus\R$ that satisfies the equation
\[
	C(z) = - \left( zI + \sum_{i = 1, i \ne j}^{n} \E \left[\Wc_{ij} C(z) \Wc_{ij}^* \right] \right)^{-1}
\]
(for any $j$ we get the same thing on the right-hand side) and such that $\Im(z)\Im C(z) \ge 0$ (entrywise).  

Observe that $\E\left[\Wc_{ij}\Wc_{ij}^*\right]=\frac{1}{n}I_{d\times d}$, so trying a solution of the form $C(z)=\alpha_n(z)I$, we obtain the functional equation
\[
	\alpha_n(z) = -\left( z+\frac{n-1}{n}\alpha_n(z) \right)^{-1} 
\]
with $\Im(z) \Im \alpha(z) > 0$. Solving, we get
\[
	\alpha_n(z) = \frac{1}{2 \frac{n-1}{n}}\left( -z - \sqrt{z^2-4\frac{n-1}{n}} \right) \,,
\]
with the point-wise limit
\[
	\alpha_n(x) \to \alpha(x) = -\frac{1}{2}\left( z - \sqrt{z^2-4} \right)
\]
being the Stieltjes transform of the semicircle law (see for example \cite{anderson2009introduction}, page 47). Thus the laws $F_n$ converge to the semicircle law, and the theorem is proved.

\begin{remark}
	Note that we didn't need to use here the fact that representation $\pi$ is irreducible; we only needed the blocks $\pi(g_{ij})$ to be unitary and centered. The proof of Theorem \ref{thm:edge}, however, relies on the exact second moments of the individual matrix elements.
\end{remark}

\subsection{Proof of Theorem \ref{thm:edge}}

Our proof is very similar the argument for the case of the extreme eigenvalues of an i.i.d Wigner matrix, see the books $\cite{tao2012topics,bai2010spectral,anderson2009introduction}$ for example. 

Observe that by Theorem \ref{thm:bulk}, we know that it must be that almost surely, 
\[
\lim\inf \lambda_{\max}(\Wc) \ge 2,\,\,\text{and }\,\, \lim\sup \lambda_{\min}(\Wc) \le -2 \,.
\]
To show this formally, simply take any positive, smooth function $\alpha$ supported on (say) $[2-\epsilon,2]$. Since $f_{sc}$ is strictly positive on $(-2,2)$, by Theorem \ref{thm:edge} we have that
\[
\frac{1}{nd} \tr\left( \alpha(\Wc) \right) := \frac{1}{nd} \sum_{i=1}^{nd} \alpha\left(\lambda_{i}(\Wc) \right) \aslim \int_{-2}^{2}\alpha(\lambda)f_{sc}(\lambda)d\lambda > 0 \,.
\]
This means, in particular, that asymptotically almost surely, the interval $[2-\epsilon,2]$ contains $O(nd)$ eigenvalues of $\Wc$. Now, the fact that the limiting density of eigenvalues is zero outside of $[-2,2]$ means that any compact interval outside of it contains at most strictly $o(nd)$ eigenvalues. It remains to show, then, that asymptotically almost surely there are no outlying eigenvalues outside the support of the semicircle law, at all. Since $\lambda_{\max}(\Wc),-\lambda_{\min}(\Wc) \le \norm{\Wc}$, it clearly suffices to show that almost surely, $\lim\sup \norm{\Wc} \le 2$.

We shall obtain a tail bound on $\norm{\Wc}$ using high-order moments. For any $k\ge 1$, clearly, $\norm{\Wc}^{2k} \le \tr(\Wc^{2k})$, and when $k=\omega(\log(nd))$, it is easy to see that $\norm{\Wc} \sim \left( \tr(\Wc^{2k}) \right)^{1/2k}$ \footnote{Indeed, $\norm{\Wc} \le \left( \tr(\Wc^{2k}) \right)^{1/2k} \le (nd)^{1/2k}\norm{\Wc}=(1+o(1))\norm{\Wc}$, since $k=\omega(\log(nd))$.}. Hence, one expects (and we shall see that this indeed turns out to be the case) that the tail bound (obtained by taking the power and using Markov's inequality)
\[
\Pr\left( \norm{\Wc} \ge 2+t \right) \le (2+t)^{-2k}\E\left[ \tr\left(\Wc^{2k}\right) \right]
\]
where $k$ is large enough, indeed captures the true behavior of $\norm{\Wc}$.  Estimating the high-order moment $\E\left[ \tr\left(\Wc^{2k}\right) \right]$ will require a rather sophisticated combinatorial calculation. Fortunately, most of the estimates we will need can be imported almost verbatim from the existing proof for the i.i.d Wigner case (that is, $d=1$ dimensional blocks).

\paragraph{Truncation step.}
Instead of working directly with the matrix $\Wc$, we will work with a modified version where all the entries are guaranteed to be bounded with probability $1$. 
Let $\overline{\Ec}$ be the matrix obtained by truncating all the entries of $\Ec$ to a magnitude at most $B\sqrt{\log(nd)}$, where $B>0$ is a number to be chosen later, that is,
\[
\overline{\Ec}_{ij} = \Ec_{ij}\cdot \I\left[\abs{\Ec_{ij}} \le B\sqrt{\log(nd)} \right] \,.
\]
We use a rather crude estimate,
\begin{align*}
	\E\norm{\Ec-\overline{\Ec}} 
	&\le \E\norm{\Ec-\overline{\Ec}}_F \\
	&\le \sqrt{\E\norm{\Ec-\overline{\Ec}}_F^2} \\
	&\le \sqrt{  \sum_{i=1}^{nd} \sum_{i=j}^{nd}  \E \left[ \abs{\Ec_{ij}}^2\cdot \I\left[\abs{\Ec_{ij}} > B\sqrt{\log(nd)} \right] \right] } \\
	&=   O\left(nd \cdot   e^{-\Omega(B^2\log(nd))} \right) \,,
\end{align*}
where the big-Oh/Omega notation hides purely numerical constants, that may differ between the two cases where the entries of $\Ec$ are real or complex standard Gaussians. 
Let
\[
\overline{\Wc} =\frac{1}{\tau_n\sqrt{n}} E \odot\left(  \overline{\Delta} \odot \Pi + \frac{\sigma_n}{\sqrt{d}}\overline{\Ec} \right)
\]
(obtained by replacing $\Ec$ with $\overline{\Ec}$ in the definition of $\Wc$), so that
\[
\E\norm{\Wc-\overline{\Wc}} \le  \frac{q_n \sigma_n}{\tau_n \sqrt{nd}}\E\norm{\Ec-\overline{\Ec}} =   O\left(  \sqrt{nq_n d}\cdot e^{-\Omega(B^2\log(nd))} \right)
\]
(recall that $\tau_n \ge \sigma_n \sqrt{q_n}$), so that for an appropriate choice of $B>0$ we can ensure that this is $\le n^{-2}$, say. Hence by Markov's inequality and the Borel-Cantelli lemma, $\norm{\Wc-\overline{\Wc}} \aslim 0$, so that it is now enough to show that $\lim\sup\norm{\Wc} \le 2$ almost surely. Note that due to the sub-Gaussian nature of $\Ec$, we didn't need to work almost at all for this truncation step. In the case where one can only assume a weaker moment bound on the additive noise (e.g, existence of a fourth moment), truncation can be somewhat more involved (see, for example, \cite{tao2012topics}). 

\paragraph{Estimating the $2k$-th moment.}
We denote the elements of $\overline{\Wc}$ by $W_{ij}$ (for brevity we omit the overline). Expanding the trace above,
\[
\E \left[ \tr (\overline{\Wc}^{2k})\right] = \sum_{s_1,\ldots,s_{2k}=1}^{nd} \E \left[ W_{s_1 s_2} \cdots W_{s_{2k-1}s_{2k}} W_{s_{2k}s_1} \right] 
\]
we can group the indices $s_1,\ldots,s_{2k}$ according to the blocks to which they correspond. That is, we can write
\begin{equation}
\E \left[ \tr (\overline{\Wc}^{2k})\right] = \sum_{i_1,\ldots,i_{2k}=1}^{n} M_{i_1,\ldots,i_{2k}}
\end{equation}
where 
\begin{equation}
M_{i_1,\ldots,i_{2k}} = \sum_{j_1,\ldots,j_{2k}=1}^{d} \E \left[ W_{d(i_1-1)+j_1,d(i_2-1)+j_2} \cdots W_{d(i_{2k}-1)+j_{2k},d(i_1-1)+j_1} \right] 
\end{equation}
is a sum over all the indices such that $(j_l,j_{l+1})$ belongs to block $(i_l,i_{l+1})$. 

As is usual in this sort of moment calculations, we can identify the index tuple $i_1,\ldots,i_{2k}$ with a directed cycle $C : \, i_1 \to \ldots \to i_{2k} \to i_{1}$ on a vertex set $V \subset [n]$. We say that $C$ contains an (undirected) edge $(i,j)$ if the (directed) cycle contains either of $i\to j$ or $j \to i$. Since the blocks of $\Wc$ are independent (unless they are in symmetric positions with respect to the diagonal) and since the matrix $\Wc$ is centered, observe that it must be that $M_{i_1,\ldots,i_{2k}}=0$ unless every edge in $i_1,\ldots,i_{2k}$ occurs at least twice, that is, the terms  
\[
\E \left[ W_{s_1 s_2} \cdots W_{s_{2k} s_1} \right] 
\]
vanish unless they contain either $0$ or at least $2$ representatives of every block (or its symmetric pair). We now proceed to bound $M_{i_1,\ldots,i_{2k}}$ when this is the case.

\begin{lemma}
	\label{lemma:symbol_bound}
	Suppose that every edge in $i_1,\ldots,i_{2k}$ occurs at least twice. Let $b$ be the number of unique edges in the cycle. Then
	\begin{equation}
	\abs{M_{i_1,\ldots,i_{2k}}} \le dn^{-k} \cdot \left(  \frac{d^5+Bd^{4.5}\sigma_n \sqrt{\log(nd)}}{\tau_n }\right)^{2k-2b} \,.
	\end{equation}
\end{lemma}
\begin{proof}
	Let $a_1,\ldots,a_b$ be the multiplicities by which each one of the unique edges occurs in $i_1,\ldots,i_{2k}$ (ordered, say, by first appearance). We also let $l$ be the number of such edges that have multiplicity strictly larger than $2$,
	\[
	l = \abs{ \left\{ s \,:\, a_s \ge 3 \right\} } \,,
	\]
	so that $b-l$ is the number of edges of multiplicity $2$. 
	Observe that to account for all $2k$ directed legs in the cycle, we must have 
	\[
	2(b-l) + \sum_{s:a_s>2} a_s = 2k \,.
	\]
	and using $\sum_{s:a_s>2} a_s \ge 3l$ we get
	\[
	l \le 2(k-b) \,. 
	\]
	By the orthogonality relations, Eq. \eqref{eq:schur}, if $(s_l,s_{l+1})$ and $(s_{k},s_{k+1})$ belong to the same block that is represented exactly twice, the term $\E \left[ W_{s_1 s_2} \cdots W_{s_{2k} s_1} \right] $ possibly doesn't vanish only if $s_l=s_k$ and $s_k=s_{k+1}$ (if they belong to symmetric blocks, we need that $s_l=s_{k+1}$, $s_{l+1}=s_k$). Let us now traverse the list $s_1\to \ldots \to s_{2k}$ and count the number of degrees of freedom we have in choosing indices to get potentially non-vanishing terms of the form
	$
	\E \left[ W_{s_1 s_2} \cdots W_{s_{2k} s_1} \right]
	$
	where $s_l=d(i_l-1)+j_l$ for $j=1,\ldots,d$. Call an index forced if it must equal an index which we already traversed; otherwise, call it unforced. 
	\begin{enumerate}
		\item The first index $s_1$ is unforced.
		\item Suppose we reach an index $t_l$ for $l>1$. If it appears in any element $W$ that belongs to a block that is represented at least $3$ times, in the worst case it is unforced. Note that there are at most $2 \sum_{s:a_s>2}a_s$ such indices.
		\item Suppose now that every appearance of $s_l$ is in a block that is represented twice. If it was already declared forced, we move on. Otherwise, its first appearance in the term $W_{s_1 s_2} \cdots W_{s_{2k}s_1}$
		is as a column-index of some element of $Z$. Suppose that $W_{s_{l-1} s_l}$ is the first representative of its block (or symmetric pair) and let $W_{s_{v-1},s_v}$ be the second (where $v>l$). If the two representatives belong to the same block, we must have that $s_v=s_l$, so that $s_v$ is forced. Otherwise, $s_v=s_{l-1}$, which is again forced. Thus, the only unforced indices $s_l$ of this type must appear as column indices in the \emph{first} representative of a block pair, so their number is $\le b-l$, the number of unique edges that appear exactly twice.
	\end{enumerate}
	Thus, we have at most $1+2\sum_{s:a_s>2}a_s+(b-l)$ unforced indices, making the number of non-vanishing terms at most 
	\[
	d^{1+2\sum_{s:a_s>2}a_s+(b-l)}=d^{1+2(2k-2(b-l))+(b-l)} = d^{1+4k-3b+3l} \,.
	\]
	Using $3l \le 6(k-b)$, we we find that there are at most 
	$d^{1+10(k-b)+b}$
	non-vanishing terms.

	Now, recall that $\Wc$ was normalized so that the absolute second moment of its entries is $1/(nd)$. The truncation we performed clearly cannot increase the absolute second moment, so this is also true for $\overline{\Wc}$. Moreover, the entries of $\overline{\Wc}$ are clearly bounded by $(1+B\sigma_n \sqrt{\log(nd)}/\sqrt{d})/(\tau_n \sqrt{n})$, where $B$ is the numerical constant we use for the truncation. Hence,
	\[
	\E \left[ W_{s_1 s_2} \cdots W_{s_{2k} s_1} \right]  \le (nd)^{-b} \left(  \frac{\sqrt{d}+B\sigma_n \sqrt{\log(nd)}}{\tau_n \sqrt{nd}}\right)^{2k-2b}\,,
	\]
	so that
	\[
	\abs{M_{i_1,\ldots,i_{2k}}} \le d^{1+10(k-b)+b} \cdot (nd)^{-b} \left(  \frac{\sqrt{d}+B\sigma_n \sqrt{\log(nd)}}{\tau_n \sqrt{nd}}\right)^{2k-2b} \,,
	\]
	from which we obtain the claimed bound.
	
\end{proof}

	 Equipped with Lemma \ref{lemma:symbol_bound}, we are ready to conclude the computation. 

	 Let $i_1 \to \ldots \to i_{2k} \to i_1$ be a (directed) cycle where each edge appears exactly twice. What is the largest number of unique vertices $v$ that such a cycle can traverse? Since the graph is connected, $v \le e+1=k+1$ ($e$ being the number of (undirected) edges), where equality holds if and only if the (undirected) graph is a tree. A cycle of this type is called \emph{non-crossing}. It is a standard calculation (see, for example, Theorem 2.3.21 in \cite{tao2012topics}) that the number of such non-crossing cycles is
	 \[
	 	C_k n(n-1)\cdots (n-k) \le 2^{2k}n^{k+1} \,,
	 \]
	 (with $C_k = \frac{1}{k+1}\binom{2k}{k}$ being the $k$-th Catalan number). 
	 Thus, by Lemma \ref{lemma:symbol_bound} (taking $b=k$), 
	 \begin{equation*}
	 \abs{ \sum_{ \substack{i_1,\ldots,i_{2k}\\\text{non-crossing}} } M_{i_1,\ldots,i_{2k}} } \le d\cdot 2^{2k}n	
	 \end{equation*}
	 It remains to bound the contribution of the rest of the cycles.

	 Let $N_{n,2k,b}$ be the number of crossing cycles of length $2k$ on $n$ vertices, that contain exactly $b=1,\ldots,k$ unique edges. Then, by Lemma \ref{lemma:symbol_bound}, clearly, 
	 \begin{equation*}
	 	\abs{ \sum_{ \substack{i_1,\ldots,i_{2k}\\\text{crossing}} } M_{i_1,\ldots,i_{2k}} } \le \sum_{b=1}^{k} dn^{-k} \cdot \left(  \frac{d^5+Bd^{4.5}\sigma_n \sqrt{\log(nd)}}{\tau_n }\right)^{2k-2b} N_{n,2k,b} \,.
	 \end{equation*}
	 In the proof of Theorem 2.3.21 in the book \cite{tao2012topics}, it is shown (the calculation is originally due to Bai and Yin, \cite{bai1988necessary}) that assuming that $k=O(\log^2(n))$,
	 \begin{equation}
	 \label{eq:bai-yin-estimate}
	 	N_{n,2k,b} \le \begin{cases}
	 		2^{2k}(2k)^{O(1)}n^{k}, &b=k \\
	 		2^{2k}(2k)^{O(2k-2b)} (2k)^{O(1)} n^{b+1} &b<k		
	 	\end{cases}
	 \end{equation}
	 where the big-Oh hides universal constants. Plugging this estimate,
	 \begin{align*}
	 	\abs{ \sum_{ \substack{i_1,\ldots,i_{2k}\\\text{crossing}} } M_{i_1,\ldots,i_{2k}} } \le 2^{2k} (2k)^{O(1)} dn^{-k+1}\cdot   \sum_{b=1}^{k} (2k)^{O(2k-2b)}  n^{b} \left(  \frac{d^5+Bd^{4.5}\sigma_n \sqrt{\log(nd)}}{\tau_n }\right)^{2k-2b}  \,.
	 \end{align*}
	 (Note that in the case $b=k$ we're losing a factor of $n$ in the last term only, but this is inconsequential , since the crossing cycles already contribute $O(2^{2k}n^{k+1})$, so that overall we have not lost anything significant here). Note that this is a geometric sum, with quotient
	 \[
	 Q = n \cdot k^{-2c} \cdot  \left(  \frac{d^5+Bd^{4.5}\sigma_n \sqrt{\log(nd)}}{\tau_n }\right)^{-2} \,,
	 \]
	 where $c$ is the constant hidden in the estimate Eq. \eqref{eq:bai-yin-estimate}. Looking at the proof in \cite{tao2012topics}, one can verify that $c>1$, but otherwise we do not attempt to optimize it, or even give precise bounds. Take $k\sim\log^{1.01}(n)$ (any $\log$ power strictly bigger than $1$ would also work). Then as long as 
	 \[
	 n \tau_n^2 \gg \log^{2.02c} \left(d^5 + Bd^{4.5}\sigma_n \log(nd)\right)^2 \,,
	 \]
	 we have that $Q \gg 1$. Since $d$, $B$ are constants, an equivalent condition is that $n\tau_n^2 \gg \log^{2.02c}(n)$ and $n\tau_n^2 \gg \sigma_n^2 \log^{2+2.02c}(nd)$. Since $\tau_n^2 \ge q_n \sigma_n^2$, the second condition could be safely replaced by $nq_n \gg \log^{2+2.02c}(nd)$, which is what we require in the statement of the Theorem. Note also that when $p_n=1$, the first condition here is not necessary (since we can discard the constant $d^5$ term we have here, which we got by suboptimaly bounding elements of $\overline{\Wc}$, in this case).
	  Under this condition, since $Q\gg 1$, the geometric sum can be bound by a constant times its largest (last, $b=k$) summand. Hence, we conclude that
	  \begin{equation}
	  \label{eq:thm-edge-moment-est}
	  	\E\tr\left(\overline{\Wc}^{2k}\right) = O\left(d2^{2k} k^{O(1)}n\right),\,\,\,\text{for } k\sim \log^{1.01}(n) \,.
	  \end{equation}
	  
	  \paragraph{Finishing the proof of Theorem \ref{thm:edge}.}
	  For any fixed $t>0$, for $k\sim \log^{1.01}(n)$ per Eq. \eqref{eq:thm-edge-moment-est},
	  \begin{align*}
	  \Pr \left(\norm{\overline{\Wc}} \ge 2+t\right) 
	  &\le (2+t)^{-2k}\E\left[\tr(\overline{\Wc}^{2k})\right] 	  	 \\
	  &= O\left( \left(\frac{2}{2+t}\right)^{2k} dk^{O(1)}n\right) \,, 
	  \end{align*}
	  which, when $k\sim \log^{1.01}(n)$ drops faster than any power of $n$. Hence $\lim\sup \norm{\overline{\Wc}} \le 2$ almost surely. Since $\norm{\Wc-\overline{\Wc}}\aslim 0$, this means that also $\lim\sup \norm{{\Wc}} \le 2$. Theorem \ref{thm:edge} now follows, as we have explained in the beginning of this section.

\section{Additional proofs}\label{sec:proofs}

	In this section we prove some of the technical claims made in the analysis of the previous sections. 

\subsection{Proof of Proposition \ref{prop:signal_noise_decouple}}
\label{subsection:signal_noise_decouple_proof}

The proof of the Proposition will follow immediately from the following Lemma:

\begin{lemma}
	Suppose that $M$ is an $nd\times nd$ symmetric block matrix, such that each block above the diagonal is all-ones with probability $t_n$ and all-zeros with probability $1-t_n$ (and, say, the block-diagonal is always all-ones). Let $A=A_n$ be sequence of $nd\times nd$ matrices that are independent of $M$, and such that $\norm{A}_{\max} = O(1)$ (here $\norm{A}_{\max}$ is the maximal entry of $A$ in magnitude) with probability $1$. 
	There is a universal numerical constant $c_0>1$, such that if 
	$t_n = \omega\left(\log^{c_0}(n)/n\right)$, then
	\[
	\frac{1}{n}\norm{\frac{1}{t_n}M\odot A - A } \aslim 0 \,.
	\]
\end{lemma}
\begin{proof}	
	It suffices to assume that the block-diagonal of $A$ is always zero. This is because the block diagonal of $\frac{1}{t_n}M\odot A - A$ always consists of elements bounded by $(1/t_n+1)\norm{A}_{\max}$, and therefor its operator norm is bounded by (say) $d(1/t_n+1)\norm{A}_{\max}$ (the Frobenius norm of any such block), and this is $o(n)$ by assumption. 
	It also suffices to only consider the case $d=1$, otherwise we can decompose the matrices $M$ and $A$ into a sum of $d^2$ $nd\times nd$ matrices, such that in every block there is at most one non-zero entry (and use the triangle inequality for the norm). For simplicity, we will also assume without loss of generality that $\norm{A}_{\max} \le 1$, and that $A$ is Hermitian (otherwise we can decompose $A=\frac{1}{2}\left(A+A^*\right) + \frac{1}{2i}\left(iA-iA^*\right)$, and again pay constant factors). 
	
	Consider the $n\times n$ matrix $Z=\frac{1}{n}(M/t_n-\one\one^T) \odot A$, where we will now think of $A$ as fixed. It is a random Hermitian block matrix, such that all the entries above the diagonal are independent of one another. The entries have mean $0$, and absolute second moment
	\begin{align*}
		\E \abs{ Z_{ij}}^2 
		&\le \E \abs{(M_{ij}/t_n)-1}^2/n^2 \\
		&= \left( t_n(1/t_n-1)^2 + (1-t_n)  \right)/n^2 \\
		&= (1-t_n) \cdot 1/(t_n n^2) \\
		&\le 1/(t_n n^2) \,.
	\end{align*}
	Since the second moment is strictly (asymptotically) smaller than $1/n$ (it is $o\left((n\log^{c_0}(n))^{-1}\right)$), one would expect that also $\E\norm{Z} = o(1)$. This is not immediate, however, since we have to account for the fact that $Z$ is also possibly sparse. In the calculation that follows, we will also need to use the crude bound $\abs{Z_{ij}} \le 1/(nt_n)$.
	
	We proceed as in the proof of Theorem \ref{thm:edge}, by bounding higher order moments of the form $\E\left[\tr(Z^{2k})\right]$, where $k$ is slightly large than $\log(n)$ \footnote{Note that if we hadn't reduced the problem to the case where $A$ is Hermitian, we would have needed to use $\E \left[\tr((ZZ^*)^k)\right]$ instead. This doesn't really change the calculation, at all.}. We follow the notations we had there.
	A term $\E \left[ Z_{i_1 i_2} \cdots Z_{i_{2k} i_1}\right]$
	is possibly non-vanishing only if in the cycle $i_1\to \ldots i_{2k} \to i_1$, every (undirected) edge appears at least twice (or doesn't appear at all). Denoting by $b\le k$ the number of unique edges, we have 
	\[
	\abs{ \E \left[ Z_{i_1 i_2} \cdots Z_{i_{2k} i_1}\right] } \le (n^2t_n)^{-b} (nt_n)^{-(2k-2b)} \,,
	\]
	hence
	\begin{align*}
		\E\left[\tr(Z^{2k})\right] 
		&\le \sum_{b=1}^{k} (n^2t_n)^{-b} (nt_n)^{-(2k-2b)} N_{n,2k,b}  \\
		&\le \sum_{b=1}^{k} (n^2t_n)^{-b} (nt_n)^{-(2k-2b)} 2^{2k}(2k)^{O(2k-2b)}(2k)^{O(1)}n^{b+1} \\
		&= n \sum_{b=1}^{k} (n t_n)^{-b} (nt_n)^{-(2k-2b)} 2^{2k}(2k)^{O(2k-2b)}(2k)^{O(1)} \,,
	\end{align*}
	(here $N_{n,2k,b}$ is the number of cycles of length $2k$ on $n$ vertices having exactly $b$ unique edges, and the estimate we plugged is the same one we had in the proof of Theorem~\ref{thm:edge}). 
	Now, this is a geometric sum, corresponding to a quotient bounded by
	\[
	Q = O(1)\cdot nt_nk^{-c_1} \,,
	\]
	where $c_1>1$ is some numerical constant that was hidden in the big-Oh notation. Requiring that $nt_n \to \infty$, and then choosing $k \sim (nt_n)^{0.99/c_1}$, now ensures us that $Q \gg 1$ for large $n$. In that case, the sum is comparable to the last term ($b=k$), so we can estimate
	\[
	n \sum_{b=1}^{k} (n t_n)^{-b} (nt_n)^{-(2k-2b)} 2^{2k}(2k)^{O(2k-2b)}(2k)^{O(1)} = O \left( n2^{2k}(2k)^{O(1)}(nt_n)^{-k} \right) \,.
	\]
	Of course,$	\norm{Z}^{2k} \le \tr(Z^{2k})$, hence by Markov's inequality
	\begin{align*}
		\Pr(\norm{Z} \ge \epsilon)
		&\le \epsilon^{-2k} \E\left[\tr(Z^{2k})\right] \\
		&\le nk^{O(1)} \left(O\left( \frac{1}{\epsilon^2 nt_n}\right)\right)^{k} \,,
	\end{align*}
	where recall that $k=(nt_n)^{0.99/c_1}$. Now, by the Borel-Cantelli lemma, all that remains to do is to ensure that this expression is summable for any fixed $\epsilon>0$. Indeed,
	\[
	(nt_n)^{k} = e^{k\cdot \log(nt_n)} = e^{(nt_n)^{0.99/c_1}\log(nt_n)} \,,
	\]
	so that if $nt_n = \omega\left( \log^{c_1/(0.99)}n \right)$, we get that the tail decays faster than any $1/poly(n)$. 
\end{proof}

	Equipped with the Lemma, Proposition~\ref{prop:signal_noise_decouple} now follows by taking $M=E\odot \Delta$ and $A=XX^{*}$, where $t_n=p_nq_n$ and we note that $\norm{XX^{*}}_{\max} \le 1$, since the blocks of $XX^*$ are all unitary matrices. 

\subsection{Proof of Proposition~\ref{prop:small-beta}}
\label{subsection:small-beta}

Basically, we need to retrace our steps through the proof of Theorem~\ref{thm:edge}, and figure out where things went wrong. The entire argument up until the estimate
\begin{align*}
	 \E \left[\tr(\overline{\Wc}^{2k})\right] = 2^{2k} (2k)^{O(1)} dn^{-k+1}\cdot   \sum_{b=1}^{k} (2k)^{O(2k-2b)}  n^{b} \left(  \frac{d^5+Bd^{4.5}\sigma_n \sqrt{\log(nd)}}{\tau_n }\right)^{2k-2b}
\end{align*}
goes through without any problem, where, as before, we take $k \sim \log^{1.01}(n)$. Noting that this is a geometric sum, when we had $n\tau_n^2 = \omega(\log^c(n))$ we could deduce that its quotient $Q \gg 1$, and then the rest of the argument of the previous proof goes through without any issues. The problems arise, then, when $Q \le 1$.
In that case, we can no longer use the last summand to bound the sum - we need to take the first one ($b=1$). That is, we have 
\begin{align*}
	\E \left[\tr(\overline{\Wc}^{2k})\right] 
	&= O\left( 2^{2k} (2k)^{O(1)} dn^{-k+2} \left(k^{O(1)}\cdot \frac{d^5+Bd^{4.5}\sigma_n \sqrt{\log(nd)}}{\tau_n }\right)^{2k-2} \right) \\
	&= O\left( k^{O(1)} dn \left(2k^{O(1)}\cdot \frac{d^5+Bd^{4.5}\sigma_n \sqrt{\log(nd)}}{\tau_n \sqrt{n}}\right)^{2k-2} \right) \,.
\end{align*}
Now, since we're dealing with $\beta_n \Wc$ instead of simply $\Wc$, we actually need to multiply this expression by $\beta_n^{2k}$, so that
\begin{align*}
	\E \left[\tr(\beta_n\overline{ \Wc})^{2k}\right]
	&= O\left( k^{O(1)} d\beta_n^2 n \left(\beta_n \cdot 2k^{O(1)}\cdot \frac{d^5+Bd^{4.5}\sigma_n \sqrt{\log(nd)}}{\tau_n \sqrt{n}}\right)^{2k-2} \right) \,.
\end{align*}
We now study the quantity that is being raise to the $k$-th power. First, note that $\sigma_n = o(1)$. Indeed, 
\[
\beta_n = \frac{\sqrt{n}\tau_n}{np_nq_n} \ge \frac{\sqrt{n} \sqrt{q_n}\sigma_n}{np_nq_n} \ge \frac{\sqrt{n} \sqrt{q_np_n}\sigma_n}{np_nq_n}
\]
so that $\sigma_n \le \sqrt{np_nq_n}\beta_n = o(1)$, by assumption. Also using $\beta/(\sqrt{n}\tau_n) = 1/(np_nq_n) = o(\log^{-c}(n))$, we see that
\[
\beta_n \cdot 2k^{O(1)}\cdot \frac{d^5+Bd^{4.5}\sigma_n \sqrt{\log(nd)}}{\tau_n \sqrt{n}} = o\left( \log^{-c}(n) k^{O(1)}(d^5+Bd^{4.5}\sigma_n \sqrt{\log(nd)}) \right) \,.
\]
Clearly, since we started with $k\sim \log^{1.01}(n)$, we could have initially picked $c$ as to ensure that this expression is (say) $o(1/\log(n))$ (and this choice is universal, in that it certainly doesn't need to depend on $d$). In that case, for any fixed $t>0$,
\begin{align*}
	\Pr \left( \norm{\beta_n \overline{Wc}} > t \right) 
	&\le t^{-2k} \E \left[\tr(\overline{\beta_n \Wc})^{2k}\right]  \\
	&= o_{d, B}\left( n \log^{O(1)}(n) \left(\frac{\log(n)}{\Omega(t^2)}\right)^{-\Omega(\log^{1.01}(n))} \right) \,,
\end{align*}
which clearly decays faster than any $1/poly(n)$. Hence, by Borel-Cantelli, $\norm{\beta_n \overline{\Wc}}\aslim 0$ which clearly implies that also without the truncation, $\norm{\beta_n{\Wc}}\aslim 0$.

\subsection{Verifying the incoherence condition of \cite{benaych2011eigenvalues}}
\label{subsection:BG-conditions}

As we mentioned before, the result of \cite{benaych2011eigenvalues} is stated for the case where either the low-rank signal or the noise matrix has a rotationaly invariant distribution. Upon a closer examination of their proof, it actually suffices to verify the following:

\begin{lemma}[Isotropy condition in the sense of \cite{benaych2011eigenvalues}]
	\label{lemma:isotropy_BGN}
	Let $x,y \in \C^{nd}$ be two different columns of $\frac{1}{\sqrt{n}}X$ (and therefor eigenvectors of $\frac{1}{n}XX^{*}$). Let $U\in U(nd)$ be the matrix whose columns are the eigenvectors of $\Wc$ (which is independent of $x$, $y$), and denote $u=Ux$, $v=Uy$. Let $\left( a_{k,nd} \right)_{k\le nd} $ for $n=1,2,\ldots$, be a sequence of uniformly bounded real numbers,
	\[
	\sup_{k,n} \abs{a_{k,nd}} \le B \,.
	\]
	Then, almost surely as $n\to\infty$,
	\begin{enumerate}
		\item 
		$\sum_{k=1}^{nd} a_{k,nd} u_k v^*_k \to  0$.
		\item If moreover 
		$\frac{1}{nd} \sum_{k=1}^{nd} a_{k,nd} \to l$,
		then  
		$\sum_{k=1}^{nd} \abs{u_k}^2 a_{k,nd} \to l$.
	\end{enumerate}
\end{lemma}

\begin{proof}
	This follows from the next lemma (Lemma \ref{lemma:trace}) with the choice of 
	\[
	A_n = U^{*} \textrm{diag}(a_{1,nd},\ldots,a_{nd,nd})U \,.
	\]
\end{proof}

\begin{lemma}[Isotropy in terms of traces]
	\label{lemma:trace}
	Let $A_n \in \C^{nd \times nd}$ be a sequence of Hermitian matrices with uniformly bounded operator norm, $\norm{A_n} \le B$. Then almost surely as $n\to\infty$, we have
	\begin{enumerate}
		\item $ u^* A_n u - \frac{1}{nd} \tr A_n \to 0$.
		\item $ u^* A_n v \to 0$.
	\end{enumerate}
\end{lemma}
\begin{proof}
	Observe that 
	\[
	\E \left[   u^* A_n v  \right] = 0, \quad 
	\E \left[   u^* A_n u \right] = \frac{1}{nd}\tr A_n 
	\]
	so all that remains now is to show that our bilinear forms concentrate around their expectations. We may assume that $A_n$ is positive semi-definite; otherwise we can treat separately its positive and negative parts: $A_n=A_n^+-A_n^-$. Let $u_i \in \C^{d} \simeq \R^{2d}$ be the part of $u$ that belongs to the $i$-th block. Note that $\norm{u_i} \le \sqrt{\frac{d}{n}}$. 
	Then 
	\[
	f(u_1,\ldots,u_n) =  u^* A_n u 
	\]
	is a convex function, with $\norm{\nabla f(u)} \le 2 \norm{A_n} \norm{u} \le 2B$, and therefor $2B$-Lipschitz. By Talagrand's concentration inequality for convex Lipschitz  functions, see for example, Exercise 6.5 in \cite{boucheron2013concentration},
	\[
	\Pr \left[ 
	\abs{ u^* A_n u - \E \left[   u^* A_n u \right]} > t
	\right] \le
	2e^{-O\left( \frac{t^2}{dB/n} \right)} \,.
	\]
	By the Borel-Cantelli lemma, ${ u^* A_n u - \E \left[   u^* A_n u \right]} \aslim 0 $, and so (1) is proved. As for (2), using
	\begin{equation*}
	\begin{split}
	&(u+v)^{*}A_n(u+v)=u^{*}A_n u + v^{*}A_n v + u^{*}A_n v + v^{*}A_n u \\
	&(u+iv)^{*}A_n(u+iv)=u^{*}A_n u + v^{*}A_n v + iu^{*}A_n v - iv^{*}A_n u 
	\end{split}
	\end{equation*}
	and repeating the argument above for every quadratic form separately, we get that
	\[
	u^{*}A_n v - \E \left[ u^{*}A_n v \right] = u^{*}A_n v \aslim 0\,, 
	\]
	as required.
\end{proof}


\bibliographystyle{alpha}
\bibliography{ref}

\newpage
\appendix
	\section*{Appendix: Harmonic analysis on compact groups}
	\label{sec:harmonic}

	In this appendix we provide background on 
	compact groups and their representations, which is assumed in the main text.

	\begin{definition}[Group representations]
		A unitary \emph{representation} of a group $\Gr$ is a homomorphism $\pi : \Gr \to U(\Hc)$, where $\Hc$ is some Hilbert space and $U(\Hc)$ is the group of unitary mappings on $\Hc$. 
		\begin{enumerate}
			\item A representation is said to be \emph{irreducible} if it has no proper sub-representation. That is, there is no proper subspace $\Vc \subsetneq \Hc$ such that
			\[
			\pi(\Gr)\Vc := \left\{ \pi(g)v \,:\, g\in\Gr,v\in\Vc,  \right\} \subset \Vc \,.
			\]
			Equivalently, an irreducible representation is one in which \emph{every} non-zero vector $v$ is cyclic, meaning that the orbit $\pi(\Gr)v:=\left\{\pi(g)v\,:\,g\in \Gr \right\}$ spans  the entirety of $\Hc$ (in the case where $\Hc$ is infinite-dimensional, it suffices that the span of the orbit is merely dense in $\Hc$, that is, its closure is $\Hc$). 
			\item The \emph{dimension} of a representation is the dimension of the underlying Hilbert space, $d = dim(\Hc)$. Whenever we use the term ``group representation'' in this article, we implicitly mean that it is finite-dimensional.

		\end{enumerate}
	\end{definition}

	The class of compact groups is particularly suitable to do probability on, because there one has a natural notion of a ``uniform distribution'' to work with, 
	\begin{definition}
		The \emph{(normalized) Haar measure}, $\mu$, on a compact group $\Gr$ is the unique left-invariant probability measure on $\Gr$. By left-invariance, we mean that for all $f\in L^1(\Gr)$,
		\[
		\int_{\Gr} f(hg) d\mu(g) = \int_{\Gr} f(g) d\mu(g)
		\]
		for all $h\in\Gr$. 
	\end{definition}

	Fix some orthonormal basis $e_1,\ldots,e_d$ of $\Hc$. With some abuse of notation, we use $\pi(g)$ to refer both to the image of $g$ as a unitary mapping, and to the unitary matrix corresponding to this mapping with respect to the basis $e_1,\ldots,e_d$. 

	Random block matrices coming from irreducible representations are very easily amenable to analysis. This is because the first two moments of their entries are all the same (the entries of course have infinitely many moments, since they are bounded), and different entries are always uncorrelated. Indeed, such matrices are centered (have mean $0$) and have variance $1/d$, as follows by the following:

	\begin{proposition}
		Suppose that $\pi$ is a nontrivial irreducible representation (in particular, $G\ne \left\{ e \right\}$). Then 
		\begin{equation}
			\int_{\Gr} \pi(g) d\mu(g) = 0
		\end{equation}
		(where $\pi(g) \in U(d)$ is a $d\times d$ unitary matrix, and the integral is entry-wise). 
	\end{proposition}
	\begin{proof}
		Let $0 \ne v \in \C^{d}$ be an arbitrary vector. Then for all $h \in \Gr$, 
		\[ \left( \int_{\Gr} \pi(g) d\mu(g) \right) v = \int_{\Gr} \pi(g)v d\mu(g) = \int_{\Gr} \pi(hg)v d\mu(g) = \pi(h) \int_{\Gr} \pi(g)v d\mu(g)\]
		which implies that $u = \int_{\Gr} \pi(g)v d\mu(g)$ is preserved under $\Gr$. If $d=1$, this mean that $u=0$ because for some $h\ne e$ we have $\pi(h) \ne 1$ (since the representation is non-trivial). Otherwise, either $u=0$ or $\textrm{span}(\pi(\Gr)u)$ is a $1$-dimensional invariant subspace. The latter option cannot be possible, since the representation is irreducible. 
	\end{proof}

	\begin{proposition}[Schur's orthogonality relations]\label{eq:schur}
		The matrix elements (with respect to any orthonormal basis) of an irreducible representation satisfy the following orthogonality relation:
		\begin{equation}
			\int_{\Gr} \pi(g)_{ij} \pi(g)_{lk}^{*} d\mu(g) = \frac{1}{d} \delta_{i=l} \delta_{j=k} \,.
		\end{equation}
		In other words, $\left\{ \sqrt{d} \pi(g)_{ij} \right\}_{i,j=1,\ldots,d}$ is an orthonormal set in $L^2(\Gr)$. 
	\end{proposition}
	\begin{proof}
		This is a fundamental result in representation theory. See, for example, the book \cite{folland2016course} (Theorem 5.8 in page 139). 
	\end{proof}

	To convince the reader that the class of irreducible representations indeed captures many of the instances of group synchronization encountered in practice, we provide some examples next.

	\begin{example}
		\begin{enumerate}
			\item Every $1$-dimensional group representation is irreducible. In particular, the obvious representation of $U(1) = \left\{ e^{i\theta} \,:\,\theta\in[0,2\pi] \right\}$ is irreducible. This was the example originally considered in the analysis of $\cite{singer2011angular}$. 
			\item The representations of $U(d)$ for $d \ge 1$, $SU(d)$ for $d \ge 2$ as rotation matrices on $\C^{d}$ are  irreducible (this is obvious). The same is true for $O(d)$ with $d \ge 1$ and $SO(d)$ with $d \ge 3$. 
			\begin{remark}
				Note that $SU(1)=\left\{ Id \right\}$ is trivial, and $O(1)=\left\{ \pm Id \right\} \simeq \Z_2$.
			\end{remark}
			Let's prove the claim for $\Gr=O(d)$ and $\Gr=SO(d)$. Let $v \ne 0 \in \C^{d}$ be some vector, which we can decompose as $v = a + ib$ for $a,b\in \R^{d}$. If $a\ne 0$ and there is some $\alpha \in \R$ with $b = \alpha a$, then 
			\[(1-\alpha i)v = (1+\alpha^2)a \ne 0 \] 
			and so $\textrm{span}(\Gr v)$ must contain $\R^{d}$ and therefor $\C^{d}$. The interesting case is $\dim \textrm{span}(a,b)=2$. But then we can find some rotation $A \in \Gr$ (here we use that $d \ge 2$ for $O(d)$ and $d \ge 3$ for $SO(d)$!) with $Ab=-b$ and $Aa \ne -a$, so that $v + Av \ne 0 \in \R^{d}$ which again implies that $\textrm{span}(\Gr v)$ contains $\R^{d}$. 
			\begin{remark}
				The representation of $SO(2)$ as a rotation matrix is not irreducible. In fact, since $SO(2)$ is abelian, one can show that all of its irreducible representations must have dimension $1$.  
			\end{remark}
			\item Identify $\Z_{L}$ with cyclic finite subgroups of $U(1)$, i.e, $\Z_{L} \simeq \left\{ 1,e^{i 2\pi /L},\ldots,e^{i 2\pi (L-1)/L} \right\}$. Then the action of $\Z_{L}$ on $\C$ by multiplication is a $1$-dimensional irreducible representation. 
		\end{enumerate}
	\end{example}

	We emphasize that the assumption that the group representation we chose is irreducible is only ever used in the analysis of the spectral method through the calculation of the first and second moments of the corresponding block matrices. In particular, we need that for our representation and choice of orthonormal basis $e_1,\ldots,e_d\in \Hc$, it holds that 
	\begin{enumerate}
	 	\item All the matrix elements have zero mean.
	 	\item All the elements have the same variance (which must in fact be $1/d$, since the matrix $\pi(g)$ is unitary).
	 	\item Different matrix elements are uncorrelated (orthogonal in $L^2(\Gr,\mu)$).
	 \end{enumerate} 
	It turns out, however, that conditions (2), (3) already imply that the representation in question must be irreducible. This is a basic argument in the representation theory of finite and compact groups, which we shall now present for completeness. By the Peter-Weyl theorem, every (unitary) representation of a compact group $\Gr$ decomposes into a direct sum of (finite-dimensional, which is obvious in our case but not at all in general) irreducible representations,
	\[
	\pi = \pi_1 \oplus \ldots \oplus \pi_k : G \to U(\Hc_1\oplus\ldots\oplus\Hc_{k}) \,.
	\]
	Consider the trace,
	\[
	\chi_{\pi}(g) = \tr \pi(g) = \sum_{i=1}^{k} \tr \pi_{k}(g) = \sum_{i=1}^{k} \chi_{\pi_{k}}(g)
	\]
	Its is known that for two irreducible representations $\pi$ and $\rho$, their traces satisfy the following orthogonality relation,
	\[
	\int_{\Gr} \chi_{\pi}(g)\chi_{\rho}(g) d\mu(g) = \begin{cases}
		1 \quad \text{$\pi$ and $\rho$ are unitarily equivalent} \\
		0 \quad \text{otherwise}
	\end{cases}
	\]
	(this follows from a more complete statement of Schur's orthogonality relations than what we gave before). But notice that properties (2) and (3) above imply that $\norm{\chi_{\pi}}^2_{L^2(\Gr)}=1$, so now
	\[
	1 = \norm{\chi_{\pi}}^2_{L^2(\Gr)} = \norm{\sum_{i=1}^{k} \chi_{\pi_{k}}}^2_{L^2(\Gr)} \ge \sum_{i=1}^{k}  \norm{\chi_{\pi_{k}}}^2_{L^2(\Gr)} = k
	\]
	and it must be that $k=1$, or equivalently, that $\pi$ is irreducible. 

	The above discussion doesn't imply, of course, that the moment conditions above are necessary for our main results on the spectral method to hold. These are simply sufficient conditions for Theorems \ref{thm:bulk} and \ref{thm:edge} to hold - and these two encapsulate (almost) all of the information we need to derive our asymptotic results. It is instructive to consider the following example of a group representation that is not irreducible.

	\begin{example}[The case of $SO(2)$]
		Consider the representation of $SO(2)$ as rotation matrices on $2$-dimensional space $\C^{2}$. We have a homomorphism $S^{1} \to SO(2)$,
		\begin{equation}
			e^{i\theta} \mapsto A(\theta) = \begin{bmatrix}
				\cos(\theta) & \sin(\theta) \\
				-\sin(\theta) & \cos(\theta) 
			\end{bmatrix} \,.
		\end{equation}
		Under this homomorphism, the pull-back of Haar measure of $SO(2)$ on $S^{1}$ is its corresponding Haar (uniform) measure, $\frac{1}{2\pi} d\theta$. While the matrices $A(\theta)$ are centered and indeed have the right variance $1/2$, the correlations between the diagonal and cross-diagonal elements do not vanish - they are $1/2$ and $-1/2$ respectively. While Theorem \ref{thm:bulk} does hold for block matrices of this form, our proof of Theorem \ref{thm:edge} doesn't readily adapt to this case; numerical evidence suggests, however, that this indeed holds for such block matrices.
	\end{example}

\end{document}